\documentclass[aap]{imsart}
\usepackage{amsmath, amsthm, amsbsy, amssymb, graphicx}
\usepackage{adjustbox,url}
\usepackage{subcaption}
\usepackage{pdfsync}
\usepackage[usenames,dvipsnames,table]{xcolor}

\usepackage{url, booktabs, wrapfig}
\usepackage{enumitem}
\setlist{noitemsep, topsep=0pt, leftmargin=1em, labelwidth=0em,
  parsep=0pt, partopsep=0pt}

\usepackage{bm}

\usepackage{booktabs, threeparttable}

\usepackage{comment} 

\newif\ifshowspecs
\showspecstrue
\showspecsfalse   
\ifshowspecs
\newcommand{\instructions}[1]{\textcolor{blue}{#1}}
\else
\newcommand{\instructions}[1]{}
\fi

\allowdisplaybreaks
\usepackage{color}

\usepackage{graphicx}

\usepackage{physics} 

\usepackage[colorlinks=true, citecolor=blue, urlcolor=blue]{hyperref}

\usepackage[]{lineno}
\newcommand*\patchAmsMathEnvironmentForLineno[1]{%
	\expandafter\let\csname old#1\expandafter\endcsname\csname 
	#1\endcsname
	\expandafter\let\csname oldend#1\expandafter\endcsname\csname 
	end#1\endcsname
	\renewenvironment{#1}%
	{\linenomath\csname old#1\endcsname}%
	{\csname oldend#1\endcsname\endlinenomath}}%
\newcommand*\patchBothAmsMathEnvironmentsForLineno[1]{%
	\patchAmsMathEnvironmentForLineno{#1}%
	\patchAmsMathEnvironmentForLineno{#1*}}%
\AtBeginDocument{%
	\patchBothAmsMathEnvironmentsForLineno{equation}%
	\patchBothAmsMathEnvironmentsForLineno{align}%
	\patchBothAmsMathEnvironmentsForLineno{flalign}%
	\patchBothAmsMathEnvironmentsForLineno{alignat}%
	\patchBothAmsMathEnvironmentsForLineno{gather}%
	\patchBothAmsMathEnvironmentsForLineno{multline}%
}

\allowdisplaybreaks 

\graphicspath{{./}{../figures/}}

\usepackage{lipsum}

\usepackage[compact]{titlesec}
\titleformat{\subparagraph}[runin]
{\normalfont\normalsize\it}{\thesubparagraph}{1em}{}

\usepackage[numbers,sort]{natbib}

\startlocaldefs

\newcommand{\PP}{\mathbb{P}}
\newcommand{\EE}{\mathbb{E}}
\newcommand{\RR}{\mathbb{R}}
\newcommand{\ind}{\boldsymbol{1}}
\newcommand{\bZ}{\boldsymbol{Z}}
\newcommand{\bI}{\boldsymbol{I}}

\newcommand{\bD}{\boldsymbol{D}}

\newcommand{\br}{\boldsymbol{r}}

\newcommand{\bv}{\boldsymbol{v}}
\newcommand{\bx}{\boldsymbol{x}}
\newcommand{\by}{\boldsymbol{y}}
\newcommand{\bu}{\boldsymbol{u}}
\newcommand{\boldf}{\boldsymbol{f}}
\newcommand{\bH}{\boldsymbol{H}}
\newcommand{\bs}{\boldsymbol{s}}
\newcommand{\bz}{\boldsymbol{z}}

\newcommand{\bolde}{\boldsymbol{e}}
\newcommand{\bxi}{\boldsymbol{\xi}}
\newcommand{\origin}{\boldsymbol{0}}

\newcommand{\Din}{D^\text{in}}
\newcommand{\Dout}{D^\text{out}}

\newcommand{\convas}{\stackrel{\text{a.s.}}{\longrightarrow}}
\newcommand{\convp}{\stackrel{p}{\longrightarrow}}

\newcommand{\Ein}{E^\text{in}}
\newcommand{\Eout}{E^\text{out}}
\newcommand{\Zin}{\Delta^\text{in}}
\newcommand{\Zout}{\Delta^\text{out}}

\newtheorem{Theorem}{Theorem}[section]

\newtheorem{Definition}{Definition}[section]
\newtheorem{Lemma}[Theorem]{Lemma}
\newtheorem{Remark}{Remark}[section]

\allowdisplaybreaks

\endlocaldefs

\begin{document}           

\begin{frontmatter}

\title{Random Networks with Heterogeneous Reciprocity}
\runtitle{Heterogeneous Reciprocity}

\begin{aug}
\author[A]{\fnms{Tiandong}~\snm{Wang}\ead[label=e1]{tw398@tamu.edu}}
\and
\author[B]{\fnms{Sidney}~\snm{Resnick}\ead[label=e2]{sir1@cornell.edu}}
\address[A]{Department of Statistics, Texas A\&M University, College Station, TX 77843, US\printead[presep={,\ }]{e1}}

\address[B]{School of Operations Research and Information Engineering, Cornell University, Ithaca, NY 14853, US\printead[presep={,\ }]{e2}}
\end{aug}

\begin{abstract}
Users of social networks display diversified behavior and online habits. For instance, a user’s tendency to reply to a post can depend on the user and the person posting. For convenience, we group users into aggregated behavioral patterns, focusing here on the tendency to reply to or \emph{reciprocate} messages. The reciprocity feature in social networks reflects the information exchange among users. We study the properties of a preferential attachment model with heterogeneous reciprocity levels, give the growth rate of model edge counts, and prove convergence of empirical degree frequencies to a limiting distribution. This limiting distribution is not only multivariate regularly varying, but also has the property of hidden regular variation. 
\end{abstract}

\begin{keyword}[class=MSC]
\kwd{60G70}
\kwd{60J85}
\end{keyword}

\begin{keyword}
\kwd{Multivariate regular variation}
\kwd{Hidden regular variation}
\kwd{Preferential attachment}
\kwd{Reciprocity}
\end{keyword}

\end{frontmatter}

\section{Introduction}

Social networks have grown rapidly and users are exhibiting different
interaction and behavioral patterns on platforms like Facebook and
Twitter. Reciprocity is one such pattern and helps characterize
information exchange between social network users
\citep{kilduff2006paradigm,molm2007building}. For instance, consider
the network of Facebook wall posts: a directed edge from user A to
user B is formed when user A leaves a message on the Facebook wall of
user B, and a reciprocal edge from user B to user A is created if 
and when user B replies to the message. Depending on both users'
behavioral features and networking habits (e.g. the closeness of
friendship, whether user A is broadcasting, the obsessiveness of user
B's replying habits), the probability of generating a reciprocal edge
may vary across different pairs of users, and network modeling should
include such heterogeneity.  A study in \cite{jiangetal:2015} shows
that online social networks tend to have a high proportion of
reciprocal edges, compared to other types of networks such as
biological networks, communication networks, software call graphs and
peer-to-peer networks. Due to the large and diversified user groups on
social networks, it is essential to incorporate heterogeneous
reciprocity levels into the modeling.

For the modeling of dynamic networks, the preferential attachment (PA)
model \citep{bollobas:borgs:chayes:riordan:2003,
  krapivsky:redner:2001} is an appealing starting point. This model
captures the scale-free property of complex networks, where both in-
and out-degree distributions have Pareto-like tails
\citep{resnick:samorodnitsky:towsley:davis:willis:wan:2016,
  resnick:samorodnitsky:2015,
  wan:wang:davis:resnick:2017,wang:resnick:2019}.  Recently,
\cite{wang:resnick:2021a} show that for a wide and realistic
range of model parameters,
the standard PA model in \cite{bollobas:borgs:chayes:riordan:2003}
generates networks with very small proportion of reciprocal edges,
deviating from the empirical observations in \cite{jiangetal:2015}.
This led to two consecutive studies \citep{wang:resnick:2022,
  cirkovic:wang:resnick:2022} on the theoretical properties and the
estimation of a PA model with homogeneous reciprocity, where
reciprocal edges are generated by simply flipping a two-sided coin.
To capture the heterogeneous reciprocal patterns, we extend here the
model in \cite{wang:resnick:2022} by dividing users into $K$ different
behavioral groups, and the probability of generating a reciprocal
edge from a user in group $m$ to a user in group $r$ is
$\rho_{m,r}\in (0,1)$, $m,r \in \{1,\ldots, K\}$.

We study three properties of the proposed PA model with heterogeneous
reciprocity. Under modest assumptions, we first analyze the
growth of the edge counts by identifying the almost sure limit for the
scaled number of edges emanating from and pointing to nodes of group
type $m$. Secondly, by embedding in- and out-degree sequences
into a family of multitype branching processes whose particles
  are given group labels, we prove that the empirical frequencies of
nodes with in-degree $i$ and out-degree $j$ converge to a limiting
distribution $p_{i,j}$. Third, we show that the asymptotic limiting
distribution, $p_{i,j}$, is multivariate regularly varying with limit
measure concentrating on a ray and after removing large in- and
out-degree pairs close to the concentrating ray, we also detect hidden
regular variation \cite{lindskog:resnick:roy:2014}.

The rest of the paper is organized as follows. We start with a detailed description of the PA model with heterogeneous reciprocity levels in Section~\ref{sec:model}. Section~\ref{sec:edge} studies the growth of edge counts, and in order
to ensure the convergence of scaled edge counts, sufficient
assumptions are imposed.  
Also, we extend the embedding technique in \cite{wang:resnick:2022} to a family of multitype branching processes with different group labels, and derive the limit of empirical degree frequencies in Section~\ref{sec:degree}. We then characterize the asymptotic dependence structure of large in- and out-degrees in Section~\ref{sec:mrv}, and give concluding remarks in Section~\ref{sec:conclusion}.
Technical proofs of results in Section~\ref{sec:edge} are collected in Section~\ref{sec:pf}.

\subsection{PA Model with Heterogeneous Reciprocity}\label{sec:model}

The proposed model extends the directed preferential attachment (PA) 
model studied in \cite{krapivsky:redner:2001, bollobas:borgs:chayes:riordan:2003, 
wan:wang:davis:resnick:2017,wang:resnick:2019} by amending a mechanism that generates
a reciprocal edge with probabilities depending on characteristics of the node pair being connected by the edge. Because these
characteristics can be specific to the node pair, the model incorporates heterogeneous responding patterns of users in social networks.

We now specify a growing sequence of graphs.
Let $G(n)$ be the graph after $n$ steps with $V(n)$ being the
set of nodes and $E(n)$ being the set of edges in $G(n)$.  
Attach to each node $v$ a communication type $W_v$, where $\{W_v, v\geq 1\}$ are iid random variables, independent from the growth mechanism of the graph, with
\[
\PP(W_v = r) =\pi_r,\qquad\text{for}\quad \sum_{r=1}^K \pi_r = 1.
\]
We imagine that 
when a node is born, it flips a multi-sided coin to determine its communication type. Think of
 $\pi_r$ as the percent of the users in a social network that have communication habits labeled as type $r$.
Let $W(n) := \{W_v: v\in V(n)\}$ denote the set of group types for all nodes in $G(n)$.
Throughout we assume that the communication group of
node $v$ is always observed upon its creation, and
remains unchanged afterwards.
In the rest of this paper, we  assume $G(n) = (V(n), E(n), W(n))$, for $n\ge 0$.

Denote the cardinality of a discrete set $S$ by $|S|$ and  initialize the model with graph $G(0)$, which consists of one node
(labeled as node~1) and a self-loop, with $V(0) = \{1\}$,
$|V(0)| = {1}$, $W(0) = \{W_1\}$, and $E(0) = \left\{(1,1)\right\}$. 
For each new edge $(u, v)$ with $W_u = r, W_v=m$, the reciprocity mechanism adds its reciprocal
counterpart $(v, u)$ instantaneously with probability $\rho_{m, r}\in [0,1]$, for $m,r\in \{1,2,\ldots,K\}$. 
Here $\rho_{m, r}$ measures the probability of adding a reciprocal edge from a node in group $m$ to a node in group $r$.
Note that the matrix $\boldsymbol{\rho} := (\rho_{m,r})_{m,r}$ is not
necessarily a stochastic matrix, but can be an arbitrary matrix in
$M_{K\times K}([0,1])$, the set of all $K\times K$ matrices with
entries belonging to $[0,1]$. 
Later in Section~\ref{sec:edge}, we will give particular regularity conditions on $\boldsymbol{\rho}$ to facilitate theoretical analysis.


Let $\bigl(\Din_v(n), \Dout_v(n)\bigr)$
be the in- and out-degrees of node $v\in V(n)$ in $G(n)$, and 
we use the convention that $\Din_v(n) = \Dout_v(n) = 0$ if
$v\notin V(n)$. 
Let $\delta > 0$ be an offset parameter.
The evolution of the network $G(n+1)$ from $G(n)$ is described
as follows.

\begin{enumerate}
\item With probability $\alpha \in (0, 1)$, 
add a new node $|V(n)|+1$ with a directed edge $(|V(n)|+1,v)$,
where $v\in V(n)$ is chosen with probability
\begin{equation}\label{eq:Din}
\frac{\Din_v(n)+\delta}{\sum_{v\in V(n)} (\Din_w(n)+\delta)}
= \frac{\Din_v(n)+\delta}{|E(n)|+\delta |V(n)|},
\end{equation}
and update the node set $V(n+1)=V(n)\cup\{|V(n)|+1\}$ and 
$W(n+1)=W(n)\cup\{W_{|V(n)|+1}\}$.
The new node $|V(n)|+1$ belongs to group $r$ with probability $\pi_r$.
If node $v$ belongs to group $m$, then 
a reciprocal edge $(v, |V(n)|+1)$ is added
with probability $\rho_{m, r}$. Upon reciprocation, update the edge set as
$E(n+1) = E(n)\cup \{(|V(n)|+1,v), (v,|V(n)|+1)\}$.
If the reciprocal edge is not created, set
$E(n+1) = E(n)\cup \{(|V(n)|+1,v)\}$.

\item With probability $\gamma \equiv 1-\alpha \in (0,1)$, 
add a new node $|V(n)|+1$ with a directed edge $(v,|V(n)|+1)$,
where $v\in V(n)$ is chosen with probability
\begin{equation}\label{eq:Dout}
\frac{\Dout_v(n)+\delta}{\sum_{v\in V(n)} (\Dout_v(n)+\delta)}
= \frac{\Dout_v(n)+\delta}{|E(n)|+\delta |V(n)|},
\end{equation}
and update the node set $V(n+1)=V(n)\cup\{|V(n)|+1\}$, 
$W(n+1)=W(n)\cup\{W_{|V(n)|+1}\}$.
The new node $|V(n)|+1$ belongs to group $r$ with probability $\pi_r$.
If node $v$ belongs to group $m$, then 
a reciprocal edge $(|V(n)|+1,v)$ is added
with probability $\rho_{r, m}$.
Upon reciprocation, update the edge set as 
$E(n+1) = E(n)\cup \{(v, |V(n)|+1,v), (|V(n)|+1,v)\}$.
If the reciprocal edge is not created, set $E(n+1) = E(n)\cup
\{(v,|V(n)|+1)\}$.
\end{enumerate}

Note that $|V(n)|=n+1$, $n\ge 0$,
since a new node is added at each step, and
the offset parameter $\delta$ is assumed to be the same for
both in- and out-degrees.


\section{Growth of Edge Counts}\label{sec:edge}
Heterogeneous reciprocity levels make analysis of the convergence of 
$|E(n)|/n$ more complicated compared to the homogeneous case.
In Theorem~\ref{thm:En}, we show the concentration of $|E(n)|$ around
$\EE[|E(n)|]$, and then prove the convergence of $\EE[|E(n)|]/n$,
thereby showing linear growth of the number of edges. We now
  prepare for this result.

Write 
\begin{align*}
\rho_{m\bullet} &:= \sum_{r=1}^K \rho_{m, r}\pi_r =\PP[\text{Type
  $m$ node sends reciprocal edge}],\\
\rho_{\bullet m} &:= \sum_{r=1}^K \rho_{r, m}\pi_r =\PP[\text{Type $m$
  node receives reciprocal edge}].
\end{align*}
Let $\mathcal{G}_n$ be denote the $\sigma$-field generated by observing the network evolution up to $n$ steps, i.e.
\[
\mathcal{G}_n = \sigma\left\{(V(k), E(k), W(k)): k=0,\ldots,n\right\}.
\]
For $v\in V(n)$, we have
\begin{align}
\EE^{\mathcal{G}_n}\left(\Din_v(n+1)-\Din_v(n)\right)
&= \alpha \frac{\Din_v(n)+\delta}{|E(n)|+\delta |V(n)|}
+ \gamma \sum_{m=1}^K\rho_{\bullet m}\frac{\Dout_v(n)+\delta}{|E(n)|+\delta |V(n)|}\ind_{\{W_v=m\}},
\label{eq:diff_in}\\
\EE^{\mathcal{G}_n}\left(\Dout_v(n+1)-\Dout_v(n)\right)
&= \gamma \frac{\Dout_v(n)+\delta}{|E(n)|+\delta |V(n)|}
+ \alpha \sum_{m=1}^K \rho_{m\bullet}\frac{\Din_v(n)+\delta}{|E(n)|+\delta |V(n)|}\ind_{\{W_v=m\}}.
\label{eq:diff_out}
\end{align}
Let $V_m(n) := \{v\in V(n): W_v=m\}$, and note that
\begin{align}
\sum_{v\in V(n)}& \EE^{\mathcal{G}_n}\left(\Din_v(n+1)-\Din_v(n)\right)\ind_{\{W_v=m\}}\nonumber\\
=& \PP^{\mathcal{G}_n}\bigl(E(n+1)=E(n)\cup\{(|V(n)|+1,v)\}, v\in V_m(n)\bigr)\nonumber\\
&+ \PP^{\mathcal{G}_n}\bigl(E(n+1)=E(n)\cup\{(v, |V(n)|+1), (|V(n)|+1,v)\},\nonumber\\
&\qquad\qquad v\in V_m(n),
\text{$(|V(n)|+1,v)$ is due to reciprocation.}\bigr)\nonumber\\
=& \PP^{\mathcal{G}_n}\bigl((|V(n)|+1,v) \in E(n+1)\setminus E(n), v\in V_m(n)\bigr).
\label{eq:sumEm}
\end{align}
Write $\Ein_m(n) := \{(u,v)\in E(n): W_v=m\}$, and $\Eout_m(n) := \{(v,u)\in E(n): W_v=m\}$.
We see that 
\begin{align}
\EE^{\mathcal{G}_n}&\left(|\Ein_m(n+1)|-|\Ein_m(n)|\right)\nonumber\\
&= \PP^{\mathcal{G}_n}\bigl((|V(n)|+1,v) \in E(n+1)\setminus E(n), v\in V_m(n)\bigr)\nonumber\\
& + \PP^{\mathcal{G}_n}\bigl((v, |V(n)|+1) \in E(n+1)\setminus E(n), W_{|V(n)|+1}=m, v\in V(n)\bigr),\nonumber\\
\intertext{and applying \eqref{eq:diff_in} and \eqref{eq:sumEm} gives}
&= \sum_{v\in V(n)} \EE^{\mathcal{G}_n}\left(\Din_v(n+1)-\Din_v(n)\right)\ind_{\{W_v=m\}}\nonumber\\
&+ \PP^{\mathcal{G}_n}\bigl((v, |V(n)|+1) \in E(n+1)\setminus E(n), W_{|V(n)|+1}=m, v\in V(n)\bigr)\nonumber\\
&= \alpha \frac{|\Ein_m(n)|+\delta|V_m(n)|}{|E(n)|+\delta |V(n)|}
+\gamma\rho_{\bullet m} \frac{|\Eout_m(n)|+\delta|V_m(n)|}{|E(n)|+\delta |V(n)|}\nonumber\\
&+ \PP^{\mathcal{G}_n}\bigl((v, |V(n)|+1) \in E(n+1)\setminus E(n), W_{|V(n)|+1}=m, v\in V(n)\bigr).\label{eq:diff_Em}
\end{align}
For the third term in \eqref{eq:diff_Em}, we have that
\begin{align}
&\PP^{\mathcal{G}_n}\bigl((v, |V(n)|+1) \in E(n+1)\setminus E(n), W_{|V(n)|+1}=m, v\in V(n)\bigr)\nonumber\\
&= \PP^{\mathcal{G}_n}\bigl((v, |V(n)|+1) \in E(n+1)\setminus E(n), W_{|V(n)|+1}=m, v\in V(n),\nonumber\\
&\qquad\qquad\text{$(v, |V(n)|+1)$ is not due to reciprocation}\bigr)
\nonumber\\
&\quad + \PP^{\mathcal{G}_n}\bigl((v, |V(n)|+1) \in E(n+1)\setminus E(n), W_{|V(n)|+1}=m, v\in V(n),\nonumber\\
&\quad\qquad\qquad\text{$(v, |V(n)|+1)$ is due to reciprocation}, W_v=r,\text{ for some $r$}\bigr)
\nonumber\\
&= \gamma\pi_m +\alpha\pi_m \sum_{r=1}^K \rho_{r, m} \frac{|\Ein_r(n)|+\delta|V_r(n)|}{|E(n)|+\delta |V(n)|}.
\label{eq:diff_Em2}
\end{align}
Therefore, combining \eqref{eq:diff_Em} and \eqref{eq:diff_Em2} shows that
\begin{align}
\EE^{\mathcal{G}_n}&\left(|\Ein_m(n+1)|-|\Ein_m(n)|\right)\nonumber\\
=& \alpha \frac{|\Ein_m(n)|+\delta|V_m(n)|}{|E(n)|+\delta |V(n)|}
+\gamma\rho_{\bullet m} \frac{|\Eout_m(n)|+\delta|V_m(n)|}{|E(n)|+\delta |V(n)|}\nonumber\\
&+\gamma\pi_m +\alpha\pi_m \sum_{r=1}^K \rho_{r, m} \frac{|\Ein_r(n)|+\delta|V_r(n)|}{|E(n)|+\delta |V(n)|}.
\label{eq:Zin}
\end{align}

Following a similar argument, we use the relationship in \eqref{eq:diff_out} to obtain that
\begin{align*}
\EE^{\mathcal{G}_n}&\left(|\Eout_m(n+1)|-|\Eout_m(n)|\right)\\
&= \gamma \frac{|\Eout_m(n)|+\delta|V_m(n)|}{|E(n)|+\delta |V(n)|}
+\alpha\rho_{m\bullet} \frac{|\Ein_m(n)|+\delta|V_m(n)|}{|E(n)|+\delta |V(n)|}\\
&+ \alpha\pi_m +\gamma\pi_m \sum_{r=1}^K \rho_{m, r} \frac{|\Eout_r(n)|+\delta|V_r(n)|}{|E(n)|+\delta |V(n)|}.
\end{align*}
In addition, summing over $m$ in \eqref{eq:Zin} gives 
\begin{align}
\label{eq:DeltaEn}
\EE^{\mathcal{G}_n}&\left(|E(n+1)|-|E(n)|\right)\nonumber\\
&= 1 +  \alpha \sum_{m=1}^K\rho_{m\bullet }\frac{|\Ein_m(n)|+\delta|V_m(n)|}{|E(n)|+\delta |V(n)|}
+\gamma \sum_{m=1}^K\rho_{\bullet m}\frac{|\Eout_m(n)|+\delta|V_m(n)|}{|E(n)|+\delta |V(n)|}.
\end{align}

As remarked at the beginning of Section \ref{sec:edge},
 the proof of Theorem~\ref{thm:En} requires showing
the convergence
  of $\EE[|E(n)|]/n$ and for this it suffices to check 
$\EE[|\Ein_m(n)|]/n\to x_m$ and $\EE[|\Eout_m(n)|]/n\to y_m$,
for some $x_m, y_m \in [0,1]$,
which is implied by the approximation
$\EE^{\mathcal{G}_n}\left(|\Ein_m(n+1)|-|\Ein_m(n)|\right)\approx x_m$, and
$\EE^{\mathcal{G}_n}\left(|\Eout_m(n+1)|-|\Eout_m(n)|\right)\approx y_m$.
Therefore, we expect $x_m, y_m$, $m=1,\ldots, K$, to be solutions to the following system of equations:
\begin{align}
x_m &= \alpha \frac{x_m+\delta\pi_m}{\sum_r x_r +\delta}+\gamma \rho_{\bullet m} \frac{y_m+\delta\pi_m}{\sum_r y_r +\delta} +\gamma\pi_m 
+ \alpha\pi_m \sum_r \rho_{r, m} \frac{x_r+\delta\pi_r}{\sum_r x_r +\delta},
\label{eq:xm}\\
y_m &= \gamma \frac{y_m+\delta\pi_m}{\sum_r y_r +\delta}+\alpha \rho_{m\bullet} \frac{x_m+\delta\pi_m}{\sum_r x_r +\delta} +\alpha\pi_m 
+ \gamma\pi_m \sum_r \rho_{m, r} \frac{y_r+\delta\pi_r}{\sum_r y_r +\delta},
\label{eq:ym}
\end{align}
for $m = 1,\ldots, K$.
Summing over $m$ in \eqref{eq:xm} and \eqref{eq:ym} shows
  $\sum_{m=1}^K x_m=\sum_{m=1}^K y_m$ and from
\eqref{eq:xm} we get
\[
\sum_{m=1}^K x_m \ge \sum_{m=1}^K\left(\alpha \frac{x_m+\delta\pi_m}{\sum_r x_r +\delta} + \gamma\pi_m \right) = 1,
\]
and therefore
 $\sum_{m=1}^K y_m\ge 1$.
Also, since 
\[
\frac{x_m+\delta\pi_m}{\sum_r x_r +\delta}\le 1,\qquad\text{and}\qquad
\frac{y_m+\delta\pi_m}{\sum_r y_r +\delta}\le 1,
\]
the right hand side of \eqref{eq:xm} is upper bounded by
\[
\alpha+\gamma\rho_{\bullet m}+\gamma\pi_m +\alpha\pi_m \text{max}_r\rho_{r,m}\le 2.
\]
Similarly, the right hand side of \eqref{eq:ym} is bounded by 2. 

Lemma~\ref{lem:xy} shows that under some regularity
conditions, the system of equations  \eqref{eq:xm} and \eqref{eq:ym}
has a unique solution, 
$(x_1,\ldots,x_K, y_1,\ldots, y_K)$ in
\[
\mathcal{Z}:= \{\bz \in [0,2]^{2K}: \sum_{i=1}^K z_i \ge 1, \sum_{i=K+1}^{2K} z_i \ge 1\}.
\] 
The proof relies on the contraction mapping theorem
 \cite[Theorem~1.2.2]{Kirk:2001}. We proceed by defining a function
 $\boldf:\mathcal{Z}\mapsto \mathcal{Z}$, using the right-hand-side
 expressions in \eqref{eq:xm} and \eqref{eq:ym}. Applying the mean
 value theorem to $\boldf$, we deduce a sufficient condition for
 $\boldf$ to be a contraction by deriving an upper bound for the
 1-norm of the Jacobian matrix. Due to the linked definition of
 $(x_m)$ and $(y_m)$ in \eqref{eq:xm} and \eqref{eq:ym}, we write the
 Jacobian matrix as a block matrix, and find the upper bounds block by
 block.

\begin{Lemma}\label{lem:xy}
Let $\odot$ and $\otimes$ denote the Hadamard and Kronecker products of matrices, respectively. Denote the identity matrix in $\mathbb{R}^K$ as $\bm{I}_K$, and use $\ind_K$ to denote the vector in $\mathbb{R}^K$ with all entries being 1.
Set $\boldsymbol{\pi}:= (\pi_1,\ldots, \pi_K)$, and the reciprocity matrix, $\bm{\rho}$, whose $(i,j)$-th entry is $\rho_{i,j}$. 
Consider the four matrices:
\begin{align*}
J^*({1,1})=\alpha\left(\bm{I}_K+(\boldsymbol{\pi}\ind_K^T)\odot \bm{\rho}^T\right),
&\quad 
J^*({1,2}) = \gamma (\bm{\rho}^T\boldsymbol{\pi})\otimes \ind_K^T,
\\ 
J^*({2,1}) = \alpha (\bm{\rho}\boldsymbol{\pi})\otimes \ind_K^T,
&\quad 
J^*({2,2}) = \gamma\left(\bm{I}_K+(\boldsymbol{\pi}\ind_K^T)\odot \bm{\rho}\right),
\end{align*}
and the block matrix
\[
J^* := 
\begin{bmatrix}
J^*({1,1}) & J^*({1,2})\\
J^*({2,1}) & J^*({2,2})
\end{bmatrix}.
\]
Let $\|J^*\|_1$ denote the 1-norm of the matrix $J^*$. Then as long as
\begin{align}
\delta> \text{max}\{\|J^*\|_1-1,0\},
\label{eq:J*}
\end{align}
the system of equations in \eqref{eq:xm} and \eqref{eq:ym} has a unique solution
in $\mathcal{Z}$.
\end{Lemma}
The proof of Lemma~\ref{lem:xy} is deferred to Section~\ref{subsec:pf_lem1}.
Continuing to assume \eqref{eq:J*} holds, the next Theorem~\ref{thm:En} gives the a.s. convergence of $|E(n)|/n$, $|\Ein(n)|/n$ and $|\Eout(n)|/n$.
Define a constant
\begin{align}
\label{eq:def_C}
C_\delta = \alpha\sum_{m=1}^K\rho_{m\bullet}\frac{x_m+\delta\pi_m}{\sum_r x_r+\delta}+ \gamma\sum_{m=1}^K\rho_{\bullet m}\frac{y_m+\delta\pi_m}{\sum_r y_r +\delta},
\end{align}
and a matrix
\begin{align}
\label{eq:def_H}
\bH := \begin{bmatrix}
\frac{\alpha}{1+\delta} \left(1+\bigvee_m \rho_{m\bullet}\right) & \frac{\gamma}{1+\delta}\bigvee_m\rho_{\bullet m} & \frac{1}{1+\delta}\left(\alpha+C_\delta\right)\\
\frac{\alpha}{1+\delta}\bigvee_m \rho_{m\bullet} & \frac{\gamma}{1+\delta}\left(1+\bigvee_m\rho_{\bullet m}\right) & \frac{1}{1+\delta}\left(\gamma+C_\delta\right)\\
\frac{\alpha}{1+\delta}\bigvee_m \rho_{m\bullet} & \frac{\gamma}{1+\delta}\bigvee_m\rho_{\bullet m}
& \frac{C_\delta}{1+\delta}
\end{bmatrix}.
\end{align}
The proof of Theorem~\ref{thm:En} requires supposing
$\lambda_H$, the largest eigenvalue of $\bH$, satisfies $\lambda_H<1$ to ensure the convergence of $\EE[|E(n)|]/n$, $\EE[|\Ein(n)|]/n$ and $\EE[|\Eout(n)|]/n$. 

Define $\Zin_m(n) := |\Ein_m(n)|+\delta|V_m(n)|-n (x_m+\delta \pi_m)$, $\Zout_m(n) := |\Eout_m(n)|+\delta|V_m(n)|-n (y_m+\delta \pi_m)$, and 
$\Delta(n) := |E(n)| - n\sum_m x_m$.
We will prove Theorem~\ref{thm:En} by deriving iterative upper bounds such that
for the matrix $\bH$ given in \eqref{eq:def_H},
\[
\begin{bmatrix}
\sum_m|\EE(\Zin_m(n+1))|\\
\sum_m|\EE(\Zout_m(n+1))|\\
|\EE(\Delta(n))|
\end{bmatrix}\le
\left(\bI+\frac{1}{n}\bH\right)
\begin{bmatrix}
\sum_m|\EE(\Zin_m(n))|\\
\sum_m|\EE(\Zout_m(n))|\\
|\EE(\Delta(n))|
\end{bmatrix}.
\]
Hence, provided that all elements in $\bH$ are strictly positive, applying the
Perron-Frobenius theorem suggests restricting the largest eigenvalue of $\bH$ will be sufficient to control the growth of the expected edge counts. In addition, the assumption of positive elements in $\bH$ leads to the requirement of $\alpha,\gamma>0$, $\bigvee_m \rho_{m\bullet}>0$, and $\bigvee_m\rho_{\bullet m}>0$. 

Recall we may
interpret $\rho_{m\bullet}$ and $\rho_{\bullet m}$,  them as the  likelihood of sending and attracting
reciprocal edges for a user in group $m$, respectively. 
The assumption $\bigvee_m \rho_{m\bullet}>0$ and
$\bigvee_m\rho_{\bullet m}>0$ requires that there exists at least one
group whose probability of generating or attracting reciprocal edges
is strictly positive.  

\begin{Theorem}\label{thm:En}
Assume $\alpha,\gamma>0$, $\bigvee_m \rho_{m\bullet}>0$, $\bigvee_m\rho_{\bullet m}>0$, 
$\delta> \|J^*\|_F-1$ and $\lambda_H<1$. 
Suppose that $(x_1,\ldots,x_K,y_1,\ldots,y_K)$ is the solution to the system of equations given in \eqref{eq:xm} and \eqref{eq:ym}.
We have:
\begin{enumerate}
\item[(i)] For $m=1,\ldots,K$, $\frac{|\Ein_m(n)|}{n}\convas x_m$, and $\frac{|\Eout_m(n)|}{n}\convas y_m$.
\item[(ii)] In addition, $\frac{|E(n)|}{n}\convas \sum_{m=1}^K x_m = \sum_{m=1}^K y_m$.
\end{enumerate}
\end{Theorem}

Note that since $(x_m,y_m)$ satisfies \eqref{eq:xm} and \eqref{eq:ym}, 
the a.s. limit of $|E(n)|/n$ depends on the value of $\delta$ as well. 
This is different from the model with a homogeneous reciprocity level, where only the reciprocity parameter determines the limit of $|E(n)|/n$.

The proof of Theorem~\ref{thm:En} is in Section~\ref{subsec:pf_thmEn}.
Throughout the rest of the paper,
assume the following regularity conditions:
\begin{equation}\label{e:star}
  \alpha,\gamma>0,\quad \delta> \|J^*\|_1-1, \quad \bigvee_m
  \rho_{m\bullet}>0,\, \bigvee_m \rho_{\bullet m}>0, \quad \lambda_H<1.
  \end{equation}

\section{Growth of Degree Counts}\label{sec:degree}

We next focus on the asymptotic behavior of the joint in- and out-degree counts:
\[
N_{k,l}(n) := \sum_{v\in V(n)}\ind_{\left\{(\Din_v(n),\Dout_v(n)) = (k,l)\right\}}.
\]
We extend the homogeneous reciprocity techniques in \cite{wang:resnick:2022,
  cirkovic:wang:resnick:2022} 
to obtain the convergence of
$N_{k,l}(n)/n$ under heterogeneous reciprocity. The way forward is via
embedding the in- and out-degree
 sequences in a family of multi-type Markov 
branching processes with immigration (MBI processes). 

\subsection{Markov Branching with Immigration}\label{subsec:MBI}

To pursue count asymptotics using embedding,
we need a linked family of two-type Markov branching processes
 (\cite{athreya:ney:2004}) where each process is a Markov Branching
  with Immigration (MBI) process. See
  Section 2.1.1 of \cite{wang:resnick:2022}.
For each $ m\in \{1,\ldots,K\}$,
$$
\{\bxi_\delta(t,m)\equiv\bigl(\xi^{(1)}_\delta(t,m),\xi^{(2)}_\delta(t,m)\bigr): t\ge 0\},
$$ is an MBI process whose branching structure
  depends on $m$; this will be specified more precisely in
  \eqref{eq:pgf1} and \eqref{eq:pgf2}.
In the rest of the paper, we refer to $\bxi_\delta(\cdot,m)$ as a \emph{MBI process with group label }$m$.
The process $\bxi_\delta(\cdot,m)$ is designed to mimic evolution of in- and out-degrees of a fixed node with communication type $m$.
The general setup
of continuous-time multitype branching processes without immigration
is reviewed in \cite[Chapter V]{athreya:ney:1972} and discussions on
the MBI process are included, for instance, in
\cite{rabehasaina:2021,wang:resnick:2022}.


We assume life time parameters of $\xi^{(1)}_\delta(\cdot,m)$ and
$\xi^{(2)}_\delta(\cdot,m)$ to be $\alpha$ and $\gamma\equiv
1-\alpha$, respectively. 
For $\bs=(s_1,s_2)\in [0,1]^2$, the branching structure of $\bxi_{\delta}(\cdot, m)$ is specified through
offspring generating functions:
\begin{align}
f^{(1)}_v(\bs,m) &= \left(1-\rho_{m\bullet}\right)s_1^2 + 
\rho_{m\bullet} s_1^2s_2 ,\label{eq:pgf1}\\
f^{(2)}_v(\bs,m) &= \left(1-\rho_{\bullet m}\right)s_2^2 + 
\rho_{\bullet m} s_1 s_2^2 ,\label{eq:pgf2}
\end{align}

According to \eqref{eq:pgf1}, when a group-$m$ type I particle's
lifetime ends,
 with  
probability $1-\rho_{m\bullet}$, it splits into two group-$m$ type I particles, increasing the total number of group-$m$
type I particles by 1, and with probability $\rho_{m\bullet}$,
it splits  into two group-$m$ type I
particles and one group-$m$ type II particle. This last eventuality
increases the total numbers of group-$m$ type I and II particles both
by 1.  
Immigration events arrive following a homogeneous Poisson process with
rate $\delta>0$.  
When an immigration event happens, 
$\bxi_{\delta}(\cdot, m)$ is incremented by $(1,0)$, $(0,1)$, or
$(1,1)$ following the distribution 
\begin{align}
\label{eq:def_p0}
p_0(\bx, m) = \left(\alpha(1-\rho_{m\bullet})\right)^{\ind_{\{\bx = (1,0)\}}}\left(\gamma(1-\rho_{\bullet m})\right)^{\ind_{\{\bx = (0,1)\}}}\left(\alpha\rho_{m\bullet}+\gamma\rho_{\bullet m}\right)^{\ind_{\{\bx = (1,1)\}}}.
\end{align}
All immigrants are of group $m$, and have the same branching structure as given in \eqref{eq:pgf1} and \eqref{eq:pgf2}.
Following the discussion in Chapter V.7.2 of \cite{athreya:ney:1972},
we use \eqref{eq:pgf1}--\eqref{eq:pgf2} to obtain a matrix that helps specify the branching structure of $\bxi_\delta(\cdot,m)$:
\begin{equation}\label{eq:defA}
A_m= 
\begin{bmatrix}
\alpha & \alpha\rho_{m\bullet}\\
\gamma\rho_{\bullet m} & \gamma
\end{bmatrix},
\qquad m\in\{1,\ldots, K\}.
\end{equation}
Applying the Perron-Frobenius theorem, $A_m$ has a largest positive eigenvalue with multiplicity 1, i.e.
\begin{align}
\label{eq:lambdam}
\lambda_m = \frac{1}{2}\left(1+\sqrt{(\alpha-\gamma)^2+4\alpha\gamma\rho_{m\bullet}\rho_{\bullet m}}\right) =: \frac{1}{2}\left(1+\sqrt{D_0(m)}\right).
\end{align}
Also, the smaller eigenvalue of $A_m$ is 
\[
\lambda'_m = \frac{1}{2}\left(1-\sqrt{D_0(m)}\right),
\]
which plays an important role in the discussion of hidden regular variation in Theorem~\ref{thm:hrv}.
Let $\bv{(m)}\equiv (v^{(1)}{(m)}, v^{(2)}{(m)})$, $\bu{(m)}\equiv (u^{(1)}{(m)}, u^{(2)}{(m)})$ be the 
left and right eigenvectors associated with $\lambda_{m}$ respectively, with all coordinates strictly positive, and $\bu(m)^T\ind =1$, $\bu(m)^T\bv(m) =1$.
Applying Theorem~1 in \cite{wang:resnick:2022} gives that
there exists some finite positive random variable $Z(m)$ such that
\begin{align}\label{eq:conv_MBI}
e^{-\lambda_m t}\bxi_{\delta}(t,m)\convas Z(m) \bv(m),\qquad m\in\{1,\dots,K\}. 
\end{align}

The PA model with heterogeneous reciprocity assumes that
the communication group of a fixed node is determined at the node's creation
by flipping a
$K$-sided coin and a corresponding feature must also be present in the embedding
process.
Let $L$ be a random variable 
with 
\begin{align}\label{eq:pmfL}
\PP(L = m)=\pi_m, \qquad m=1,\ldots, K,
\end{align}
and
for $t\ge 0$, the generating function of $\bxi_{\delta}(t, L)$ is
\[
\EE\left(s_1^{\xi^{(1)}_{\delta}(t, L)}s_2^{\xi^{(2)}_{v,\delta}(t, L)}\right)
= \sum_{m=1}^K \pi_m \EE\left(s_1^{\xi^{(1)}_{\delta}(t, m)}s_2^{\xi^{(2)}_{\delta}(t, m)}\right),\qquad s_1,s_2\in [0,1].
\]
Imagine $L$ assigns a label to a MBI which remains unchanged 
throughout the evolution of $\bxi_{\delta}(\cdot,\cdot)$.
The initial value, $\bxi_\delta(0,L)$, is a random vector on $\{(0,1), (1,0), (1,1)\}$, whose distribution depends on the label $L$ and will be specified during the embedding process.
Conditioning on $(\bxi_\delta(0,L),L)$, $\bxi_\delta(\cdot,L)$ behaves like a MBI process with fixed group label $L$.
In the sequel, we refer to 
$\{\bxi_{\delta}(t,L): t\ge 0\}$ as a two-type MBI process with random group label.

\subsection{Embedding}
Let $\{L_v:v\ge 1\}$ be random variables with common pmf's as in \eqref{eq:pmfL}. 
We now explain how to embed the degree sequence
$\{(\Din_v(n),\Dout_v(n)): v\in V(n)\}_{n\ge 1}$ into a family of linked two-type MBI processes with random group labels,
 $\{\bxi_{v,\delta}(t,L_v): t\ge 0\}_{v\ge 1}$.
We assume $\{\bxi_{v,\delta}(\cdot,\cdot)\}_{v\ge 1}$ have the same parameters but possibly different initializations.

At $T_0=0$,
flip a $K$-sided coin with outcome $L_1$.  If $L_1=m_1$, which occurs
with probability $\pi_{m_1}$, 
 initiate the MBI process $\bxi_{1,\delta}(\cdot,m_1)$ with $\bxi_{1,\delta}(0, m_1)=(1,1)$. 
Let $T_1$ be the first jump time of
$\bxi_{1,\delta}(\cdot, m_1)$.
Then for $t\ge 0$, 
\begin{align*}
\PP\left(T_1>t, L_1=m_1\right)
&= \pi_{m_1}\exp\left\{-t\left(\alpha\left(\bxi^{(1)}_{1,\delta}(0,m_1)+\delta\right)+\gamma\left(\bxi^{(2)}_{1,\delta}(0,m_1)+\delta\right)\right)\right\}\\
&=  \pi_{m_1} e^{-t(1+\delta)},
\end{align*}
and  $\PP(T_1>t)  =e^{-t(1+\delta)}$. 

At $T_1$,
{flip another $K$-sided coin with outcome $L_2$} and
{start} the  process 
$\{\bxi_{2,\delta}(t-T_1, L_2):t\ge T_1\}$. The initial value of
$\bxi_{2,\delta}(0, L_2)$ is set depending on which of the following
four cases happens:
\begin{enumerate}
\item[(i)] If {$L_1=m_1$ and at time $T_1$}  the process
    $\bxi_{1,\delta}(\cdot, {m_1})$ increases by $(1,0)$,  
then we have $L_2=m_2$ with (conditional) probability $
\pi_{m_2}(1-\rho_{m_1, m_2})/(1-\rho_{m_1\bullet})$, 
and set $\bxi_{2,\delta}(0, m_2)=(0,1)$.
\item[(ii)] If the $\bxi_{1,\delta}(\cdot, L_1)$ with $L_1=m_1$ is increased by $(0,1)$, 
then we have $L_2=m_2$ with (conditional) probability $
\pi_{m_2}(1-\rho_{m_2, m_1})/(1-\rho_{\bullet m_1})$, 
and 
set $\bxi_{2,\delta}(0, m_2)=(1,0)$.
\item[(iii)] If one type I particle in $\bxi_{1,\delta}(\cdot,L_1)$ with $L_1=m_1$ splits into 
2 type I and 1 type II particles at $T_1$,
then we have $L_2=m_2$ with (conditional) probability $
\pi_{m_2}\rho_{m_1, m_2}/\rho_{ m_1\bullet}$, 
and 
set $\bxi_{2,\delta}(0, m_2)=(1,1)$.
\item[(iv)] If one type II particle in $\bxi_{1,\delta}(\cdot,L_1)$ with $L_1=m_1$ splits into 
1 type I and 2 type II particles at $T_1$,
then we have $L_2=m_2$ with (conditional) probability $
\pi_{m_2}\rho_{m_2, m_1}/\rho_{\bullet m_1}$, 
and 
set $\bxi_{2,\delta}(0, m_2)=(1,1)$.
\end{enumerate}
Therefore, we see that
\begin{align*}
&\PP(L_1=m_1, L_2=m_2, T_2-T_1>t)\\
&= \pi_{m_1} e^{-t(1+\delta)} \left(\alpha(1-\rho_{m_1\bullet})\frac{\pi_{m_2}(1-\rho_{m_1, m_2})}{1-\rho_{m_1\bullet}}+\alpha\rho_{m_1\bullet}\frac{\pi_{m_2}\rho_{m_1, m_2}}{\rho_{m_1\bullet}}\right.\\
&\left. \quad+ \gamma(1-\rho_{\bullet m_1})\frac{\pi_{m_2}(1-\rho_{m_2, m_1})}{1-\rho_{\bullet m_1}}
+ \gamma\rho_{\bullet m_1}\frac{\pi_{m_2}\rho_{m_2, m_1}}{\rho_{\bullet m_1}}
\right)\\
&= \pi_{m_1}\pi_{m_2} e^{-t(1+\delta)},
\end{align*}
which shows that $(L_1,L_2)$ are independent.
Also, $\bxi_{2,\delta}(0, L_2)$ is a 2-dimensional random vector whose generating function
satisfies
\begin{align*}
\EE\left(s_1^{\xi^{(1)}_{2,\delta}(0, L_2)}s_2^{\xi^{(2)}_{2,\delta}(0, L_2)}\middle\vert L_1\right) &=
\alpha(1-\rho_{L_1\bullet})s_2+\gamma(1-\rho_{\bullet L_1})s_1+(\alpha\rho_{L_1\bullet}+\gamma\rho_{\bullet L_1}) s_1s_2,
\end{align*}
so that for $\rho_0:= \sum_m \pi_m \rho_{m\bullet} = \sum_m \pi_m \rho_{\bullet m}$,
\begin{align}\label{eq:pgf_xi0}
\EE\left(s_1^{\xi^{(1)}_{2,\delta}(0, L_2)}s_2^{\xi^{(2)}_{2,\delta}(0, L_2)}\right)
&= \sum_{m=1}^K\pi_{m}\left(\alpha(1-\rho_{m\bullet})s_2+\gamma(1-\rho_{\bullet m})s_1+(\alpha\rho_{m\bullet}+\gamma\rho_{\bullet m}) s_1s_2\right)\nonumber\\
&= \alpha(1-\rho_0)s_2+\gamma(1-\rho_0)s_1+\rho_0 s_1s_2,
\qquad s_1,s_2\in [0,1].
\end{align}
Define $R_1:=
\ind_{\{\bxi_{1,\delta}(T_1, L_1)=\bxi_{1,\delta}(0, L_1)+(1,1)\}}$, so that
$$\PP(R_1=1|L_1=m)=\alpha\rho_{m\bullet}+\gamma\rho_{\bullet m}=1-\PP(R_1=0|L_1=m),$$
and $\PP(R_1=1) = \rho_0 =1-\PP(R_1=0)$.
Set 
\begin{align*}
\mathcal{F}_{T_1}:=& \sigma\left(\{L_k\}_{k=1,2};
                     \left\{\bxi_{k,\delta}(t-T_{k-1}, L_k):t\in
                     [T_{k-1},T_1]\right\}_{k=1,2}\right)\\
  =& {\sigma\left(\{L_k\}_{k=1,2}; \bxi_{1,\delta}(t, L_1), t\in [0,T_1];      \bxi_{2,\delta}(0, L_2)
     \right)},
     \end{align*}
     we see that $R_1$ is $\mathcal{F}_{T_1}$-measurable.

In general, for $n\ge 1$, suppose that we have initiated $n+1$ MBI processes with group labels $\{L_k:1\le k\le n+1\}$ at time $T_n$,
\begin{align}\label{eq:MBIs}
\{\bxi_{k,\delta}(t-T_{k-1}, L_k): t\ge T_{k-1}\}_{1\le k\le n+1}.
\end{align}
Define $T_{n+1}$ as the first time when one of the processes in \eqref{eq:MBIs} jumps, and let
 $J_{n+1}$ be the index of the process that jumps at $T_{n+1}$. 
Define the $\sigma$-field
\[
\mathcal{F}_{T_n}:= \sigma\left(\{L_k\}_{k=1}^{n+1};\left\{\bxi_{k,\delta}(t-T_{k-1}, L_k):t\in [T_{k-1},T_n]\right\}_{1\le k\le n+1}\right),
\] 
and
\[
R_{n+1} := \ind_{\left\{\bxi_{J_{n+1},\delta}(T_{n+1}-T_{J_{n+1}-1}, L_{J_{n+1}})=\bxi_{J_{n+1},\delta}(T_n-T_{J_{n+1}-1},L_{J_{n+1}})+(1,1)\right\}}.
\]
Then we see that $\{R_k:1\le k\le n\}$ are $\mathcal{F}_{T_n}$-measurable.

At $T_{n+1}$, we initiate the MBI process with group label $L_{n+2}$, $\{\bxi_{n+2,\delta}(t-T_{n+1}, L_{n+2}):t\ge T_{n+1}\}$, and one of the following four cases happens:
\begin{enumerate}
\item[(i)] If the $\bxi_{J_{n+1},\delta}(\cdot, L_{J_{n+1}})$ with $L_{J_{n+1}}=m_{n+1}$ is increased by $(1,0)$, 
then we have $L_{n+2}=m_{n+2}$ with (conditional) probability $
\pi_{m_{n+2}}(1-\rho_{m_{n+1}, m_{n+2}})/(1-\rho_{m_{n+1}\bullet})$, 
and set $\bxi_{n+2,\delta}(0, m_{n+2})=(0,1)$.
\item[(ii)] If the $\bxi_{J_{n+1},\delta}(\cdot, L_{J_{n+1}})$ with $L_{J_{n+1}}=m_{n+1}$ is increased by $(0,1)$, 
then we have $L_{n+2}=m_{n+2}$ with (conditional) probability $
\pi_{m_{n+2}}(1-\rho_{m_{n+2}, m_{n+1}})/(1-\rho_{\bullet m_{n+1}})$, 
and 
set $\bxi_{n+2,\delta}(0, m_{n+2})=(1,0)$.
\item[(iii)] If one type I particle in $\bxi_{J_{n+1},\delta}(\cdot, L_{J_{n+1}})$ with $L_{J_{n+1}}=m_{n+1}$ splits into 
2 type I and 1 type II particles at $T_{n+1}$,
then we have $L_{n+2}=m_{n+2}$ with (conditional) probability $
\pi_{m_{n+2}}\rho_{m_{n+1}, m_{n+2}}/\rho_{m_{n+1}\bullet}$, 
and 
set $\bxi_{n+2,\delta}(0, m_{n+2})=(1,1)$.
\item[(iv)] If one type II particle in $\bxi_{J_{n+1},\delta}(\cdot, L_{J_{n+1}})$ with $L_{J_{n+1}}=m_{n+1}$ splits into 
1 type I and 2 type II particles at $T_{n+1}$,
then we have $L_{n+2}=m_{n+2}$ with (conditional) probability $
\pi_{m_{n+2}}\rho_{m_{n+2}, m_{n+1}}/\rho_{\bullet m_{n+1}}$, 
and 
set $\bxi_{n+2,\delta}(0, m_{n+2})=(1,1)$.
\end{enumerate}
Since $\sum_{k=1}^{n+1} \xi^{(1)}_{k,\delta}(T_{n}-T_{k-1}, L_k)=\sum_{k=1}^{n+1} \xi^{(2)}_{k,\delta}(T_{n}-T_{k-1})=n+1+\sum_{k=1}^n R_k$, we then have
\begin{align*}
&\PP^{\mathcal{F}_{T_n}}\left(L_{n+2}=m, T_{n+1}-T_n>t\right)\\
&= \sum_{r=1}^K \PP^{\mathcal{F}_{T_n}}\left(L_{n+2}=m, L_{J_{n+1}}=r, T_{n+1}-T_n>t\right)\\
&= \exp\left\{-t\left((1+\delta)(n+1)+\sum_{k=1}^n R_k\right)\right\}\\
&\times
\sum_{r=1}^K \left[\frac{\sum_{k=1}^{n+1}(\xi^{(1)}_{k,\delta}(T_n-T_{k-1}, L_k)+\delta)\ind_{\{L_k=r\}}}{(1+\delta)(n+1)+\sum_{k=1}^n R_k}\left(\alpha(1-\rho_{r\bullet})\frac{\pi_{m}(1-\rho_{r, m})}{1-\rho_{r\bullet}}+\alpha\rho_{r\bullet}\frac{\pi_{m}\rho_{r, m}}{\rho_{r\bullet}}\right)
\right.\\
&\left. \quad+ \frac{\sum_{k=1}^{n+1}(\xi^{(2)}_{k,\delta}(T_n-T_{k-1}, L_k)+\delta)\ind_{\{L_k=r\}}}{(1+\delta)(n+1)+\sum_{k=1}^n R_k}\left(\gamma(1-\rho_{\bullet r})\frac{\pi_{m}(1-\rho_{m, r})}{1-\rho_{\bullet r}}
+ \gamma\rho_{\bullet r}\frac{\pi_{m}\rho_{m, r}}{\rho_{\bullet r}}\right)
\right]\\
&=\pi_m \exp\left\{-t\left((1+\delta)(n+1)+\sum_{k=1}^n R_k\right)\right\}\\
&= \PP^{\mathcal{F}_{T_n}}\left(L_{n+2}=m\right) \PP^{\mathcal{F}_{T_n}}\left(T_{n+1}-T_n>t\right).
\end{align*}
Hence, $L_{n+2}$ and $T_{n+1}-T_n$ are independent under $\PP^{\mathcal{F}_{T_n}}$, and since 
\begin{align}\label{eq:Ln}
\PP^{\mathcal{F}_{T_n}}\left(L_{n+2}=m\right) =\pi_m = \PP\left(L_{n+2}=m\right),
\end{align}
 $L_{n+2}$ is independent from $\{L_k: 1\le k\le n+1\}$.

In addition, from the four cases listed above, we also see that
\begin{align*}
&\PP^{\mathcal{F}_{T_n}}(\bxi_{n+2,\delta}(0, L_{n+2})=(0,1), L_{J_{n+1}}=m, L_{n+2}=r)\\
&=\frac{\alpha\pi_r(1-\rho_{m, r})
\sum_{k=1}^{n+1} \left(\xi^{(1)}_{k,\delta}(T_{n}-T_{k-1}, L_k)+\delta\right)\ind_{\{L_k= m\}}}{\alpha\sum_{k=1}^{n+1} \xi^{(1)}_{k,\delta}(T_{n}-T_{k-1}, L_k)
+ \gamma\sum_{k=1}^{n-1} \xi^{(2)}_{k,\delta}(T_n-T_{k-1}, L_k) +(n+1)\delta}\\
&= \alpha\pi_r(1-\rho_{m, r})\frac{\sum_{k=1}^{n+1} \left(\xi^{(1)}_{k,\delta}(T_{n}-T_{k-1}, L_k)+\delta\right)\ind_{\{L_k= m\}}}{(1+\delta)(n+1)+\sum_{k=1}^n R_k}.
\end{align*}
Similarly, we have
\begin{align*}
\PP^{\mathcal{F}_{T_n}}&(\bxi_{n+2,\delta}(0, L_{n+2})=(1,0), L_{J_{n+1}}=m, L_{n+2}=r)\\
&= \gamma\pi_r(1-\rho_{r, m})\frac{\sum_{k=1}^{n+1} \left(\xi^{(2)}_{k,\delta}(T_{n}-T_{k-1}, L_k)+\delta\right)\ind_{\{L_k= m\}}}{(1+\delta)(n+1)+\sum_{k=1}^n R_k},
\end{align*}
and
\begin{align*}
\PP^{\mathcal{F}_{T_n}}&(\bxi_{n+2,\delta}(0, L_{n+2})=(1,1), L_{J_{n+1}}=m, L_{n+2}=r)\\
=& \frac{\alpha\pi_r\rho_{m, r}\sum_{k=1}^{n+1} \left(\xi^{(1)}_{k,\delta}(T_{n}-T_{k-1}, L_k)+\delta\right)\ind_{\{L_k= m\}}}{(1+\delta)(n+1)+\sum_{k=1}^n R_k}\\
&+ \frac{\gamma\pi_r\rho_{r, m}\sum_{k=1}^{n+1} \left(\xi^{(2)}_{k,\delta}(T_{n}-T_{k-1}, L_k)+\delta\right)\ind_{\{L_k= m\}}}{(1+\delta)(n+1)+\sum_{k=1}^n R_k}.
\end{align*}
Therefore, under $\PP^{\mathcal{F}_{T_n}}$, 
$\bxi_{n+2,\delta}(0, L_{n+2})$ is a random vector following the distribution
\begin{align*}
p_{n+2,0}(\bx)
&= \left(\sum_{m=1}^K\frac{\alpha(1-\rho_{m\bullet})\sum_{k=1}^{n+1} \left(\xi^{(1)}_{k,\delta}(T_{n}-T_{k-1}, L_k)+\delta\right)\ind_{\{L_k= m\}}}{(1+\delta)(n+1)+\sum_{k=1}^n R_k}\right)^{\ind_{\{\bx = (0,1)\}}}\\
&\times \left(\sum_{m=1}^K\frac{\gamma(1-\rho_{\bullet m})\sum_{k=1}^{n+1} \left(\xi^{(2)}_{k,\delta}(T_{n}-T_{k-1}, L_k)+\delta\right)\ind_{\{L_k= m\}}}{(1+\delta)(n+1)+\sum_{k=1}^n R_k}\right)^{\ind_{\{\bx = (1,0)\}}}\\
&\times \left(
\sum_{m=1}^K\left(\frac{\alpha\rho_{m\bullet}\sum_{k=1}^{n+1} \left(\xi^{(1)}_{k,\delta}(T_{n}-T_{k-1}, L_k)+\delta\right)\ind_{\{L_k= m\}}}{(1+\delta)(n+1)+\sum_{k=1}^n R_k}\right.\right.\\
&\left.\left.\quad + \frac{\gamma\rho_{\bullet m}\sum_{k=1}^{n+1} \left(\xi^{(2)}_{k,\delta}(T_{n}-T_{k-1}, L_k)+\delta\right)\ind_{\{L_k= m\}}}{(1+\delta)(n+1)+\sum_{k=1}^n R_k}
\right)
\right)^{\ind_{\{\bx = (1,1)\}}}.
\end{align*}
Meanwhile,
since $\{L_k: 1\le k\le n+1\}$ are $\mathcal{F}_{T_n}$-measurable, we then have
\begin{align}
&\PP^{\mathcal{F}_{T_n}}\left(R_{n+1}=1, J_{n+1}=w, T_{n+1}-T_n>t\right) \nonumber\\
&= \frac{\alpha \rho_{L_w\bullet} \left(\xi^{(1)}_{w,\delta}(T_n-T_{w-1}, L_w)+\delta
\right)+\gamma \rho_{\bullet L_w}\left(\xi^{(2)}_{w,\delta}(T_n-T_{w-1}, L_w)+\delta
\right)}{\alpha\sum_{k=1}^{n+1} \xi^{(1)}_{k,\delta}(T_n-T_{k-1}, L_k)
+ \gamma\sum_{k=1}^{n+1} \xi^{(2)}_{k,\delta}(T_n-T_{k-1}, L_k) +(n+1)\delta}\nonumber\\
&\quad\times\exp\left\{-t\left(\alpha\sum_{k=1}^{n+1} \xi^{(1)}_{k,\delta}(T_n-T_{k-1}, L_k)
+ \gamma\sum_{k=1}^{n+1} \xi^{(2)}_{k,\delta}(T_n-T_{k-1}, L_k) +(n+1)\delta\right)\right\}\nonumber\\
&= \frac{\alpha \rho_{L_w\bullet} \left(\xi^{(1)}_{w,\delta}(T_n-T_{w-1}, L_w)+\delta
\right)+\gamma \rho_{\bullet L_w}\left(\xi^{(2)}_{w,\delta}(T_n-T_{w-1}, L_w)+\delta
\right)}{(n+1)(1+\delta)+\sum_{k=1}^n R_k} e^{-t((n+1)(1+\delta)+\sum_{k=1}^n R_k)}\nonumber\\
&= \PP^{\mathcal{F}_{T_n}}\left(R_{n+1}=1, J_{n+1}=w\right)
\PP^{\mathcal{F}_{T_n}}\left(T_{n+1}-T_n>t\right).
\label{eq:RnJn}
\end{align}
Therefore, $\bigl(R_{n+1},J_{n+1}\bigr)$ is independent from $T_{n+1}-T_n$ under 
$\PP^{\mathcal{F}_{T_n}}$.

Define for $n\ge 0$,
\[
{\bxi}^*_\delta\left(T_n, L_{[n+1]}\right):=\left(\bxi_{1,\delta}(T_n, L_1),\bxi_{2,\delta}(T_n-T_1, L_2),\ldots, \bxi_{n+1,\delta}(0, L_{n+1}),
(0,0),\ldots\right), 
\]
then the embedding framework just described shows that
$\{{\bxi}^*_\delta\left(T_n, L_{[n+1]}\right): n\ge 0\}$ is Markov with state space $\left(\mathbb{N}^2\right)^\infty$.
The next theorem gives the embedding of in- and out-degree sequences in the PA model with heterogeneous reciprocity into a linked system of delayed MBI processes.
\begin{Theorem}\label{thm:embed_MBI}
In $\left(\mathbb{N}^2\right)^\infty$, define the in- and out-degree sequences as
\[
\bD(n) := \left(\bigl(\Din_1(n),\Dout_1(n)\bigr),\ldots, \bigl(\Din_{n+1}(n),\Dout_{n+1}(n)\bigr),
(0,0),\ldots\right)\qquad n\ge 0.
\]
Then for $\{T_k: k\ge 0\}$ and $\{\bxi_{k,\delta}(t-T_{k-1}, L_k): t\ge T_{k-1}\}_{k\ge 1}$ constructed above,
we have that in $\left(\left(\mathbb{N}^2\right)^\infty\right)^\infty$,
\begin{align*}
\bigl\{\bD(n): n\ge 0\bigr\} \stackrel{d}{=} \left\{{\bxi}^*_\delta\left(T_n, L_{[n+1]}\right): n\ge 0\right\}.
\end{align*}
\end{Theorem}
\begin{proof}
By the model description in Section~\ref{sec:model}, $\{\bD(n):n\ge 0\}$ is Markovian on $\left(\mathbb{N}^2\right)^\infty$, so
it suffices to check the transition probability from $\bD(n)$ to $\bD(n+1)$ agrees with that from 
${\bxi}^*_\delta\left(T_n, L_{[n+1]}\right)$ to ${\bxi}^*_\delta\left(T_{n+1}, L_{[n+2]}\right)$.
Write
\begin{align*}
\bolde_v^{(1)} &:= \left(\bigl(0,0\bigr),\ldots,\bigl(0,0\bigr),\underbrace{\bigl(1,0\bigr)}_{\text{$v$-th entry}},\bigl(0,0\bigr),\ldots\right),\\
\bolde_v^{(2)} &:= \left(\bigl(0,0\bigr),\ldots,\bigl(0,0\bigr),\underbrace{\bigl(0,1\bigr)}_{\text{$v$-th entry}},\bigl(0,0\bigr),\ldots\right),\\
\bolde_v^{(3)} &:= \left(\bigl(0,0\bigr),\ldots,\bigl(0,0\bigr),\underbrace{\bigl(1,1\bigr)}_{\text{$v$-th entry}},\bigl(0,0\bigr),\ldots\right),
\end{align*}
and we have
\begin{align}
\PP^{\mathcal{G}_n}\left(\bD(n+1)=\bD(n)+\bolde_v^{(1)} + \bolde_{|V(n)|+1}^{(2)}\right)
&= \frac{\alpha(\Din_v(n)+\delta)}{|E(n)|+\delta |V(n)|}\sum_{m=1}^K(1-\rho_{W_v, m}\pi_m)\nonumber\\
&= \frac{\alpha(\Din_v(n)+\delta)}{|E(n)|+\delta |V(n)|}(1-\rho_{W_v\bullet});\label{eq:PAtrans1}\\
\intertext{similarly,}
\PP^{\mathcal{G}_n}\left(\bD(n+1)=\bD(n)+\bolde_v^{(2)}+ \bolde_{|V(n)|+1}^{(1)}\right)
&= \frac{\gamma(\Dout_v(n)+\delta)}{|E(n)|+\delta |V(n)|}(1-\rho_{\bullet W_v}),\label{eq:PAtrans2}\\
\PP^{\mathcal{G}_n}\left(\bD(n+1)=\bD(n)+\bolde_v^{(3)}+ \bolde_{|V(n)|+1}^{(3)}\right)
&= \frac{\alpha\rho_{W_v\bullet}(\Din_v(n)+\delta)+\gamma\rho_{\bullet W_v}(\Dout_v(n)+\delta)}{|E(n)|+\delta |V(n)|}.
\label{eq:PAtrans3}
\end{align}
Note that $|V(n)|=n+1$ for all $n\ge 0$, and 
from \eqref{eq:RnJn}, we have
\begin{align}
\PP^{\mathcal{F}_{T_n}}&\left({\bxi}^*_\delta\left(T_{n+1}, L_{[n+2]}\right)={\bxi}^*_\delta\left(T_n, L_{[n+1]}\right)+\bolde_v^{(3)}+ \bolde_{|V(n)|+1}^{(3)}\right)\nonumber\\
&= \PP^{\mathcal{F}_{T_n}}\left(R_{n+1}=1, J_{n+1}=v\right)\nonumber\\
&=\frac{\alpha \rho_{L_v\bullet} \left(\xi^{(1)}_{v,\delta}(T_n-T_{v-1}, L_v)+\delta
\right)+\gamma \rho_{\bullet L_v}\left(\xi^{(2)}_{v,\delta}(T_n-T_{v-1}, L_v)+\delta
\right)}{(n+1)(1+\delta)+\sum_{k=1}^n R_k}.
\label{eq:trans_embed}
\end{align}
So it remains to check
whether $\sum_{k=1}^n R_k$ has the same distribution as $|E(n)|-(n+1)$.

Again, applying \eqref{eq:RnJn} gives that
\begin{align*}
\PP^{\mathcal{F}_{T_{n}}}&\left(R_{n+1}=1\right)\\
&= \sum_{v=1}^n \frac{\alpha \rho_{L_v\bullet} \left(\xi^{(1)}_{v,\delta}(T_n-T_{v-1}, L_v)+\delta
\right)+\gamma \rho_{\bullet L_v}\left(\xi^{(2)}_{v,\delta}(T_n-T_{v-1}, L_v)+\delta
\right)}{(n+1)(1+\delta)+\sum_{k=1}^n R_k}.
\end{align*}
Also, we obtain from \eqref{eq:Zin} that $|E(n)|-(n+1)$ satisfies
\begin{align}\label{eq:En_prob}
\PP^{\mathcal{G}_n}&\left(|E(n+1)|-|E(n)|-1=1\right)\nonumber\\
&= \alpha\sum_{m=1}^K \rho_{m\bullet}\frac{|\Ein_m(n)|+\delta|V_m(n)|}{|E(n)|+\delta(n+1)}
+\gamma \sum_{m=1}^K \rho_{\bullet m}\frac{|\Eout_m(n)|+\delta|V_m(n)|}{|E(n)|+\delta(n+1)}
\nonumber\\
&= \sum_{v=1}^n \frac{\alpha \rho_{W_v\bullet} \left(\Din_v(n)+\delta
\right)+\gamma \rho_{\bullet W_v}\left(\Dout_v(n)+\delta
\right)}{|E(n)|+\delta (n+1)}.
\end{align}
By \eqref{eq:Ln}, we see that $\{L_v: v\ge 1\}$ are iid random variables with 
$\PP(L_v = m)=\pi_m$, agreeing with the distributional property of $\{W_v: v\ge 1\}$.
Therefore, we conclude from \eqref{eq:En_prob} that $\sum_{k=1}^n R_k$ has the same distribution as $|E(n)|-(n+1)$, which implies the agreement between the transition probabilities in 
\eqref{eq:PAtrans3} and \eqref{eq:trans_embed}.
\end{proof}

\subsection{Asymptotic Growth of Empirical Degree Frequencies}

We now use the embedding results in Theorem~\ref{thm:embed_MBI} to prove the convergence of 
$N_{k,l}(n)/n$.

\begin{Theorem}\label{thm:limitNij} 
Let $\{\widetilde{\bxi}_\delta(t, L^*):t\ge 0\}$ be a MBI process with
random group label $L^*$, where $L^*$ 
satisfies $\PP(L^*=m)=\pi_m$, $m=1,\ldots,K$.
Suppose that the regularity condition \eqref{e:star} holds, and
 the initialization,
$\widetilde{\bxi}_\delta(0, L^*)$, satisfies that for
$s_i\in [0,1]$, $i=1,2$,
\begin{align}
\EE&\left(s_1^{\widetilde{\xi}^{(1)}_\delta(0, L^*)}s_2^{\widetilde{\xi}^{(2)}_\delta(0, L^*)}\right)\nonumber\\
&= \sum_{r=1}^K \pi_r \left[\alpha \left(1-\sum_{m=1}^K \rho_{m, r}\frac{x_m+\delta\pi_m}{\sum_s x_s +\delta}\right) s_2 + \gamma \left(1-\sum_{m=1}^K \rho_{r, m}\frac{y_m+\delta\pi_m}{\sum_s y_s +\delta}\right) s_1\right.\nonumber\\
&\left.+ \left(\alpha \sum_{m=1}^K \rho_{m, r}\frac{x_m+\delta\pi_m}{\sum_s x_s +\delta}
+\gamma \sum_{m=1}^K \rho_{r, m}\frac{y_m+\delta\pi_m}{\sum_s y_s +\delta}\right)s_1s_2
\right].\label{eq:initial}
\end{align}
For $L^*=m$, the branching structure of $\{\widetilde{\bxi}_\delta(t, m):t\ge 0\}$ is given by $A_m$ (cf. \eqref{eq:defA}).
Write 
\begin{equation}\label{e:rho*}
  \rho^*:= \sum_{m=1}^K x_m-1 = \sum_{m=1}^K y_m-1>0   \text{
    and } 
  c^*:= 1+\rho^*+\delta, 
  \end{equation}

then as $n\to\infty$, we have
for $k,l\ge 0$,
\begin{align}\label{eq:Nij}
\frac{N_{k,l}(n)}{n}&\convp \sum_{m=1}^K \pi_m\int_0^\infty c^* e^{-c^*t}
\PP\left(\widetilde{\bxi}_\delta(t, m)=(k,l)\right)\dd t. 
\end{align}
\end{Theorem}

Let $T^*$ be an exponential random variable with rate $c^*$, independent from 
$\{\widetilde{\bxi}_\delta(\cdot,m): m=1,\ldots,K\}$, then the integral on the right hand side of \eqref{eq:Nij} is
\[
\PP\left(\widetilde{\bxi}_\delta(T^*, m)=(k,l)\right),\qquad m=1,\ldots, K,
\]
representing the limiting in- and out-degree frequencies for nodes of communication group $m$.
In the sequel, set $(\mathcal{I}_m, \mathcal{O}_m)=
\widetilde{\bxi}_\delta(T^*, m)$ to denote {the} limit random
variables so that 
\eqref{eq:Nij} becomes
\begin{align}
\label{eq:Nij_IO}
\frac{N_{k,l}(n)}{n}&\convp 
\sum_{m=1}^K \pi_m\PP\Bigl[\big(\mathcal{I}_m,\mathcal{O}_m\bigr) = (k,l)\Bigr].
\end{align}

\begin{proof}
By the embedding results in Theorem~\ref{thm:embed_MBI}, we have
\begin{align}\label{eq:Nij_dist}
\frac{N_{k,l}(n)}{n}\stackrel{d}{=} \frac{1}{n}\sum_{w=2}^{n+1}\ind_{\left\{\bxi_{w,\delta}(T_n-T_{w-1},L_w)=(k,l)\right\}}
+ \frac{1}{n}\ind_{\left\{\bxi_{1,\delta}(T_n, L_1)=(k,l)\right\}},
\end{align}
where the second term goes to 0 a.s. as $n\to\infty$.
Then we only need to consider the limit of the first term in \eqref{eq:Nij_dist}.

Let $\{\widetilde{\bxi}_{w,\delta}(t, \widetilde{L}_w): t\ge 0\}_{w\ge 1}$ be a sequence of iid MBI processes with random labels, which have the same distributional properties as $\widetilde{\bxi}_{\delta}(\cdot, L^*)$ and satisfy the initialization condition in \eqref{eq:initial}. 
Then we divide the first term in \eqref{eq:Nij_dist} into different parts:
\begin{align*}
&\frac{1}{n}\sum_{w=2}^{n+1}\ind_{\left\{\bxi_{w,\delta}(T_n-T_{w-1}, L_w)=(k,l)\right\}}\\
&= 
\frac{1}{n}\sum_{w=2}^{n+1}\left(\ind_{\left\{\bxi_{w,\delta}\left(T_n-T_{w-1},L_w\right)=(k,l)\right\}}
- \ind_{\left\{\bxi_{w,\delta}\bigl(\log(n/w)/c^*,L_w\bigr)=(k,l)\right\}}
\right)\\
&\,+ 
\frac{1}{n}\sum_{w=2}^{n+1}\left(
\ind_{\left\{\bxi_{w,\delta}\bigl(\log(n/w)/c^*\bigr)=(k,l)\right\}}
- \PP^{\mathcal{F}_{T_{w-1}}}\left[\bxi_{w,\delta}\left(\frac{1}{c^*}\log(n/w), L_w\right)=(k,l)\right]
\right)\\
&\,+ 
\frac{1}{n}\sum_{w=2}^{n+1}\left(\PP^{\mathcal{F}_{T_{w-1}}}\left[\bxi_{w,\delta}\left(\frac{1}{c^*}\log(n/w), L_w\right)=(k,l)\right]
\right.\\
&\left. \qquad \qquad \qquad - \PP\left[\widetilde{\bxi}_{w,\delta}\left(\frac{1}{c^*}\log(n/w), \widetilde{L}_w\right)=(k,l)\middle| \widetilde{\bxi}_{w,\delta}(0, \widetilde{L}_w), \widetilde{L}_w\right]
\right)\\
&\,+\left(
\frac{1}{n}\sum_{w=2}^{n+1}\PP\left[\widetilde{\bxi}_{w,\delta}\left(\frac{1}{c^*}\log(n/w), \widetilde{L}_w\right)=(k,l)\middle| \widetilde{\bxi}_{w,\delta}(0, \widetilde{L}_w), \widetilde{L}_w\right]\right.\\
&\left. \qquad - \int_0^1 \PP\left[\widetilde{\bxi}_{\delta}\left(-\frac{1}{c^*}\log t,L^*\right)=(k,l)\right]\dd t
\right)\\
&\, + \int_0^1 \PP\left[\widetilde{\bxi}_{\delta}\left(-\frac{1}{c^*}\log t,L^*\right)=(k,l)\right]\dd t\\
&=: A_1(n)+A_2(n)+A_3(n)+A_4(n)+A_5.
\end{align*}
Note that by a change of variable argument, $A_5$ is identical to the right hand side of \eqref{eq:Nij},
and we now show that $A_j(n)\convp 0$, for $j=1,2,3,4$.

For $A_1(n)$, we have
\begin{align}
\EE|A_1(n)|&\le \frac{1}{n}\sum_{w=2}^{n+1}\EE\left|\ind_{\left\{\xi^{(1)}_{w,\delta}(T_n-T_{w-1}, L_w)=k\right\}}
-\ind_{\{\xi^{(1)}_{w,\delta}\bigl(\log(n/w)/c^*, L_w\bigr)=k\}}\right|\nonumber\\
&\quad +\frac{1}{n}\sum_{w=2}^{n+1}\EE\left|\ind_{\{\xi^{(2)}_{w,\delta}(T_n-T_{w-1}, L_w)=l\}}
-\ind_{\{\xi^{(2)}_{w,\delta}\bigl(\log(n/w)/c^*, L_w\bigr)=l\}}\right|.
\label{eq:A1_bound}
\end{align}
Since both $\xi^{(1)}_{w,\delta}(\cdot, L_w)$ and $\xi^{(2)}_{w,\delta}(\cdot, L_w)$
 have finite number of jumps in any finite time interval $[0,K]$ a.s., then 
 applying Lemma~3.1 in \cite{athreya:ghosh:sethuraman:2008} gives that for all $K>0$ and $w\ge 2$,
 \[
 \lim_{\epsilon\downarrow 0}\sup_{t\in [0,K]} \PP\left(\xi^{(i)}_{w,\delta}(t+\epsilon, L_w)-
 \xi^{(i)}_{w,\delta}\bigl((t-\epsilon)\wedge 0, L_w)\bigr)\ge 1\right)=0,\qquad i=1,2.
 \]
 Also, we see from \cite[Corollary 2.1(iii)]{athreya:ghosh:sethuraman:2008} that for $\eta>0$,
 \[
 \sup_{n\eta \le w\le n}\left|T_n-T_{w-1}-\frac{1}{c^*}\log(n/w)\right|\convas 0.
 \]
 Then using {techniques from} \cite[Theorem 1.2, pp 489--490]{athreya:ghosh:sethuraman:2008} further gives
 \begin{align*}
 &\left|\ind_{\{\xi^{(1)}_{w,\delta}(T_n-T_{w-1}, L_w)=k\}}
-\ind_{\{\xi^{(1)}_{w,\delta}\bigl(\log(n/w)/c^*, m\bigr)=k\}}\right|\\
&\le \sup_{w\ge 1}\sup_{1\le m\le K}\sup_{t\in [0,-\log\eta/c^*]}\PP\left(\xi^{(1)}_{w,\delta}(t+\epsilon, m)-
 \xi^{(1)}_{w,\delta}\bigl((t-\epsilon)\wedge 0, m)\bigr)\ge 1\right)\\
&\quad +\PP\left(\sup_{n\eta \le w\le n}\left|T_n-T_{w-1}-\frac{1}{c^*}\log(n/w)\right|\ge \epsilon\right)=: p_1(\epsilon,\eta).
 \end{align*}
 Similarly,
 \begin{align*}
 &\left|\ind_{\{\xi^{(2)}_{w,\delta}(T_n-T_{w-1}, L_w)=l\}}
-\ind_{\{\xi^{(2)}_{w,\delta}\bigl(\log(n/w)/c^*, L_w\bigr)=l\}}\right|\\
&\le \sup_{w\ge 1}\sup_{1\le m\le K}\sup_{t\in [0,-\log\eta/c^*]}\PP\left(\xi^{(2)}_{w,\delta}(t+\epsilon,m)-
 \xi^{(2)}_{w,\delta}\bigl((t-\epsilon)\wedge 0), m\bigr)\ge 1\right)\\
&\quad +\PP\left(\sup_{n\eta \le w\le n}\left|T_n-T_{w-1}-\frac{1}{c^*}\log(n/w)\right|\ge \epsilon\right)=: p_2(\epsilon,\eta).
 \end{align*}
 Then by \eqref{eq:A1_bound}, we see that
 \begin{align*}
 \EE|A_1(n)|&\le 2\cdot\frac{1}{n}\cdot n\eta+\frac{1}{n}(1-\eta)n\bigl(p_1(\epsilon,\eta)+p_2(\epsilon,\eta)\bigr),
 \end{align*}
 which implies $\lim_{n\to\infty}\EE|A_1(n)|=0$. Therefore, $A_1(n)\convp 0$.
 
  For $A_2(n)$, 
define 
$$X_w:= \ind_{\left\{\bxi_{w,\delta}\left(\frac{1}{c^*}\log(n/w), L_w\right)=(k,l)\right\}}
 - \PP^{\mathcal{F}_{T_{w-1}}}\left(\bxi_{w,\delta}\left(\frac{1}{c^*}\log(n/w), L_w\right)=(k,l)\right),$$ 
 then we see that $\EE(X_w) = \EE\left(\EE^{\mathcal{F}_{T_{w-1}}}(X_w)\right)=0$.
 Also, for $u<w$, since
 \[
 \EE(X_wX_u) = \EE\left(\EE^{\mathcal{F}_{T_{w-1}}}(X_w X_u)\right)
 =\EE\left(\EE^{\mathcal{F}_{T_{w-1}}}(X_w) \EE^{\mathcal{F}_{T_{w-1}}}(X_u)\right)=0,
 \]
then by the weak law of large numbers, we have $A_2(n)\convp 0$.

For $A_3(n)$, we first note that for $w\ge 2$ and $t\ge 0$, 
\[
\PP^{\mathcal{F}_{T_{w-1}}}\left(\bxi_{w,\delta}(t, L_w) = (k,l)\right)
= \PP\left(\bxi_{w,\delta}(t, L_w) = (k,l)\middle| \bxi_{w,\delta}(0, L_w), L_w\right).
\]
Then we have
\begin{align*}
|A_3(n)|&\le \frac{1}{n}\sum_{w=2}^{n+1}\ind_{\left\{\left(\bxi_{w,\delta}(0,L_w),L_w\right)\neq \left(\widetilde{\bxi}_{w,\delta}(0,\widetilde{L}_w),\widetilde{L}_w\right)\right\}}.
\end{align*}
Hence, to prove $A_3(n)\convp 0$, it suffices to show
\begin{align}\label{eq:xi0_diff}
\frac{1}{n}\sum_{w=2}^{n+1}\PP\left[\left(\bxi_{w,\delta}(0,L_w),L_w\right)\neq \left(\widetilde{\bxi}_{w,\delta}(0,\widetilde{L}_w),\widetilde{L}_w\right)\right] \to 0.
\end{align}
Let $\mathcal{X} := \{(0,1),(1,0),(1,1)\}\times \{1,\ldots, K\}$, and we note that
\begin{align*}
\PP&\left[\left(\bxi_{w,\delta}(0,L_w),L_w\right)\neq \left(\widetilde{\bxi}_{w,\delta}(0,\widetilde{L}_w),\widetilde{L}_w\right)\right]\\
&\le \sum_{\bx \in \mathcal{X}}\left|\PP\left[\left(\bxi_{w,\delta}(0,L_w),L_w\right)=\bx\right]-\PP\left[\left(\widetilde{\bxi}_{w,\delta}(0,\widetilde{L}_w),L_w\right)=\bx\right]\right|\\
&\le 2\alpha\sum_{m=1}^K \rho_{m\bullet}\left|\EE\left(\frac{|\Ein_m(w-1)|+\delta|V_m(w-1)|}{|E(w-1)|+\delta |V(w-1)|}\right)-\frac{x_m+\delta\pi_m}{\sum_s x_s +\delta}\right|\\
&\quad + 2\gamma\sum_{m=1}^K \rho_{\bullet m}\left|\EE\left(\frac{|\Eout_m(w-1)|+\delta|V_m(w-1)|}{|E(w-1)|+\delta |V(w-1)|}\right)-\frac{y_m+\delta\pi_m}{\sum_s y_s +\delta}\right|.
\end{align*}
From the proof of Theorem~\ref{thm:En}, we see that there exist constants $C_1(\delta), C_2(\delta)>0$ such that
\[
\left|\EE\left(\frac{|\Ein_m(w-1)|+\delta|V_m(w-1)|}{|E(w-1)|+\delta |V(w-1)|}\right)-\frac{x_m+\delta\pi_m}{\sum_s x_s +\delta}\right|\le C_1(\delta) w^{\lambda_H-1},
\]
and
\[
\left|\EE\left(\frac{|\Eout_m(w-1)|+\delta|V_m(w-1)|}{|E(w-1)|+\delta |V(w-1)|}\right)-\frac{y_m+\delta\pi_m}{\sum_s y_s +\delta}\right|\le C_2(\delta) w^{\lambda_H-1}.
\]
Then the left hand side of \eqref{eq:xi0_diff} is bounded by
\[
\frac{2}{n}\sum_{w=2}^{n+1} \sum_{m=1}^K\left(\rho_{m\bullet} C_1(\delta) 
+ \rho_{\bullet m} C_2(\delta)\right) w^{\lambda_H-1}, \qquad \lambda_H<1,
\]
which goes to 0 as $n\to\infty$, thus proving the claim in \eqref{eq:xi0_diff}.

 For $A_4(n)$, since $\{(\widetilde{\bxi}_{w,\delta}(0,\widetilde{L}_w), \widetilde{L}_w): w\ge 1\}$ are iid 
 random vectors in $\mathcal{X}$, then 
\begin{align*}
 \frac{1}{n}\sum_{w=2}^{n+1}&\left(\PP\left[\widetilde{\bxi}_{w,\delta}\left(\frac{1}{c^*}\log(n/w), \widetilde{L}_w\right)=(k,l)\middle| \widetilde{\bxi}_{w,\delta}(0, \widetilde{L}_w), \widetilde{L}_w\right]
 \right.\\
&\left.\qquad - \PP\left[\widetilde{\bxi}_{w,\delta}\left(\frac{1}{c^*}\log(n/w), \widetilde{L}_w\right)=(k,l)\right]\right)\convp 0.
\end{align*}
Also, by the definition of $\widetilde{\bxi}_{w,\delta}$, $w\ge 1$, we see that
\[
\frac{1}{n}\sum_{w=2}^{n+1}\PP\left[\widetilde{\bxi}_{w,\delta}\left(\frac{1}{c^*}\log(n/w), \widetilde{L}_w\right)=(k,l)\right]
= \frac{1}{n}\sum_{w=2}^{n+1}\PP\left[\widetilde{\bxi}_{\delta}\left(\frac{1}{c^*}\log(n/w), L^*\right)=(k,l)\right].
\]
Since the function $\PP[\widetilde{\bxi}_{\delta}(t,L^*)=(k,l)]$ is bounded and continuous in $t$, 
 then we conclude that $A_4(n)\convp 0$ by applying the Riemann integrability of $\PP\left[\widetilde{\bxi}_{\delta}(-\log t/c^*, L^*)=(k,l)\right]$,
which completes the proof of \eqref{eq:Nij}.
\end{proof}

\section{Power Laws and Asymptotic Dependence of Degree Frequencies}\label{sec:mrv}
In this section, we study the  dependence between large in- and
out-degrees by examining the asymptotic behavior of the distribution
$\PP[(\mathcal{I}_m, \mathcal{O}_m)\in \cdot]$, for each
$m=1,\ldots,K$. 

\subsection{Multivariate and Hidden Regular Variation}

To formalize our analysis, we provide some useful definitions related
to \emph{multivariate regular variation} (MRV) and \emph{hidden
  regular variation} (HRV) {of distributions}.

Suppose that $\mathbb{C}_0\subset\mathbb{C}\subset\mathbb{R}_+^2$ are two closed cones, and we provide the definition of $\mathbb{M}$-convergence in Definition~\ref{def:Mconv}
(cf. \cite{lindskog:resnick:roy:2014,hult:lindskog:2006a,das:mitra:resnick:2013,kulik:soulier:2020,basrak:planinic:2019}) on $\mathbb{C}\setminus \mathbb{C}_0$, which lays the theoretical foundation of regularly varying measures (cf. Definition~\ref{def:MRV}).
\begin{Definition}\label{def:Mconv}
Let $\mathbb{M}(\mathbb{C}\setminus \mathbb{C}_0)$ be the set of Borel
measures on $\mathbb{C}\setminus \mathbb{C}_0$ which are finite on
sets bounded away from $\mathbb{C}_0$, and
$\mathcal{C}(\mathbb{C}\setminus \mathbb{C}_0)$ be the set of
continuous, bounded, non-negative functions on $\mathbb{C}\setminus
\mathbb{C}_0$ whose supports are bounded away from  $\mathbb{C}_0$. Then for $\mu_n,\mu \in \mathbb{M}(\mathbb{C}\setminus
\mathbb{C}_0)$, we say $\mu_n \to \mu$ in
$\mathbb{M}(\mathbb{C}\setminus \mathbb{C}_0)$, if $\int f\dd
\mu_n\to\int f\dd \mu$ for all $f\in \mathcal{C}(\mathbb{C}\setminus
\mathbb{C}_0)$. 
\end{Definition} 

Without loss of generality \cite{lindskog:resnick:roy:2014}, we can
and do take functions in  $\mathcal{C}(\mathbb{C}\setminus
\mathbb{C}_0)$ to be uniformly continuous as well.  Denote the modulus
of continuity of a uniformly continuous function $f:\RR_+^p \mapsto
\RR_+$ by
\begin{equation}\label{e:defModCon}
\Delta_f(\delta)=\sup\{ |f(\bx)-f(\by)|:
d(\bx,\by)<\delta\}\end{equation}
 where $d(\cdot,\cdot)$ is an appropriate metric
on the domain of $f$. Uniform continuity means $\lim_{\delta \to 0}
\Delta_f(\delta)=0.$ 
{We now present the definition of multivariate regular variation with
$\mathbb{C}=\mathbb{R}_+^2$ and $\mathbb{C}_0 = \{\origin\}$.

Following Definition~2.1 in \cite{resnickbook:2007}, 
we denote a regularly varying function $f:\RR_+\mapsto \RR_+$ with index $a\in\RR$, as $f\in RV_a$.
Definition~\ref{def:MRV} gives the formal description of the MRV of distributions.

\begin{Definition}\label{def:MRV}
The distribution $\PP(\bZ\in\cdot)$ of a
 random vector $\bZ$ on $\mathbb{R}_+^2$,
  is (standard) regularly varying on $\mathbb{R}_+^2\setminus
  \{\origin\}$ with index $c>0$  if
  there exists some {regularly varying} scaling function $b(t)\in \text{RV}_{1/c}$ and a
  limit measure $\nu(\cdot)\in \mathbb{M}(\mathbb{R}_+^2\setminus
  \{\origin\})$ such that 
as $t\to\infty$,
\begin{equation}\label{eq:def_mrv}
t\PP(\bZ/b(t)\in\cdot)\rightarrow \nu(\cdot),\qquad\text{in }\mathbb{M}(\mathbb{R}_+^2\setminus \{\origin\}).
\end{equation}
It is  convenient to write  $\PP(\bZ\in\cdot)\in  
  \text{MRV}(c, b(t), \nu, \mathbb{R}_+^2\setminus \{\origin\})$.
\end{Definition}}

When analyzing the asymptotic dependence between components of a
bivariate random vector
$\bZ$ satisfying \eqref{eq:def_mrv},  it is often informative  
to make a polar coordinate transform and
consider the transformed points located on the $L_1$ unit sphere
\begin{align}
\label{eq:map_L1}
(x,y)\mapsto\left(\frac{x}{x+y},\frac{y}{x+y}\right),
\end{align}
after thresholding the data according to the
$L_1$ norm. 

When a limit measure
concentrates on a subcone of the full state space,
to improve estimates of probabilities in the complement of the subcone,
we can seek a
second {\it hidden\/} regular variation regime after removing the
subcone. 

\begin{Definition}
The vector $\bZ$ is regularly varying on $\RR^2_+\setminus \{\origin\}$ and has hidden regular variation on $\RR^2_+ \setminus\mathbb{C}_0$ if there exist $0 <c\le c_0$, scaling functions $b(t) \in RV_{1/c}$ and $b_0(t) \in RV_{1/c_0}$ with $b(t)/b_0(t) \to\infty$ and limit measures $\nu$, $\nu_0$, such that 
\begin{align}\label{eq:def_hrv}
\PP(\bZ\in\cdot)\in 
\text{MRV}(c, b(t), \nu, \mathbb{R}_+^2\setminus \{\origin\})
\cap \text{MRV}(c_0, b_0(t), \nu_0, \mathbb{R}_+^2\setminus \{\mathbb{C}_0\}).
\end{align}
\end{Definition}

A convenient way to characterize HRV is through the \emph{generalized polar coordinate transformation} for $\mathbb{R}_+^2\setminus \mathbb{C}_0$ and an associated metric $d(\cdot,\cdot)$ satisfying $d(cx,cy) = cd(x,y)$ for scalars $c>0$. The metric $d(\cdot,\cdot)$ that we use in practice is the $L_1$-metric.
When using generalized polar coordinates with respect to the forbidden
zone $\mathbb{C}_0$, we define $\aleph_{\mathbb{C}_0} :=\{\bx
\in\mathbb{C}\setminus \mathbb{C}_0 : d(\bx,\mathbb{C}_0) = 1\}$, the
locus of points at distance 1 from 
$\mathbb{C}_0$. 
Then the generalized polar coordinates are specified through the transformation,
$\text{GPOLAR}: \RR_+^2\setminus \mathbb{C}_0\mapsto (0,\infty)\times \aleph_{\mathbb{C}_0}$ with
\[
\text{GPOLAR}(\bx) = \left(d(\bx,\mathbb{C}_0),\frac{\bx}{d(\bx,\mathbb{C}_0)}\right).
\]
Let $\nu_c(\cdot)$ be a measure in $\mathbb{M}(\RR_+\setminus\{0\})$ satisfying $\nu_c(x,\infty)=x^{-c}$ $x,c>0$, and $S_0(\cdot)$ be a probability measure on $\aleph_{\mathbb{C}_0}$. 
Then generalized polar coordinates allow re-writing \eqref{eq:def_hrv} as
\[
t\PP\left[\left(\frac{d(\bZ, \mathbb{C}_0)}{b_0(t)}, \frac{\bZ}{d(\bZ, \mathbb{C}_0)}\right)\in\cdot\right]
\to (\nu_{c_0}\times S_0)(\cdot)
\]
in $\mathbb{M}((\RR_+\setminus\{0\})\times \aleph_{\mathbb{C}_0})$.
See \cite{das:mitra:resnick:2013} and \cite{lindskog:resnick:roy:2014} for details.

\subsection{Degree Frequencies and HRV}

Let $L^*$ be a random variable with pmf $\PP(L^*=m)=\pi_m$, $m=1,\ldots,K$, independent from $T^*$ and $\{\widetilde{\bxi}_\delta(\cdot,m): m=1,\ldots,K\}$.
Using \eqref{eq:Nij},  the limit in \eqref{eq:Nij_IO} becomes
\[
\sum_{m=1}^K \pi_m\PP\Bigl[\big(\mathcal{I}_m,\mathcal{O}_m\bigr) =
(k,l)\Bigr]=\PP\left(\widetilde{\bxi}_\delta(T^*, L^*) =
  (k,l)\right)=:{\PP \Big(\bigl(\mathcal{I},\mathcal{O}\bigr) =(k,l) \Bigr)},
\]
the limiting 
empirical proportion of nodes with in-degree $k$ and out-degree $l$.

We now discuss the regular variation properties of this limit distribution.
The branching structure of $\widetilde{\bxi}_\delta(\cdot,m)$ is specified through the matrix $A_m$ given in \eqref{eq:defA}, and
the largest eigenvalue of $A_m$ is $\lambda_m$ given in \eqref{eq:lambdam}.
We assume parameters $\alpha$, $\gamma$ and the matrix
$\boldsymbol{\rho}$, are chosen such that $\lambda_m$, $m=1,\ldots,K$,
are all distinct values {and}  without loss of generality, assume the
behavioral group labels are chosen so that
\begin{equation}\label{e:ordered}
\lambda_{1}>\lambda_{2}>\cdots>\lambda_{K},\end{equation}
Theorem~\ref{thm:mrv}
gives the multivariate regular variation
of $\PP[(\mathcal{I},\mathcal{O}) \in \cdot\,]$ on $\mathbb{R}_+^2
\setminus \{\origin\}$. The proof is based on
an extended Breiman's theorem  \cite[Theorem 
3]{wang:resnick:2022}, reviewed in Theorem~\ref{th:extendBrei}. 
Theorem~\ref{thm:mrv} requires an additional assumption that $\lambda_1\ge \log 2$ to guarantee that the moment condition in \eqref{e:extraCond} is satisfied.

\begin{Theorem}\label{th:mrvIO}
\label{thm:mrv}
Recall the definition of $\rho^*,\, c^*$ in \eqref{e:rho*},
 suppose $\lambda_{1}\ge \log 2$ and that the regularity conditions
\eqref{e:star} hold. 
Then as $t\to\infty$,
\[
t\PP\left[\frac{(\mathcal{I},\mathcal{O})}{t^{\lambda_{1}/c^*}}\in\cdot\right] \longrightarrow \mu_1,
\qquad\text{in}\quad\mathbb{M}(\RR_+^2\setminus \{\boldsymbol{0}\}),
\]
where the limiting measure $\mu_1\in
\mathbb{M}(\mathbb{R}_+^2\setminus\{\origin\})$ satisfies that for $f\in \mathcal{C}(\mathbb{R}_+^2\setminus\{\origin\})$,
\begin{align*}
\mu_1(f) = &\PP(L^*=1)\int_0^\infty
             \EE\left(f(y\widetilde{Z}(1)\bv(1))\right) 
             \nu_{c^*/\lambda_{1}}(\mathrm{d}y).\\
  =&\PP(L^*=1) \EE \left( \widetilde{Z}(1)^{c^*/\lambda_1}\right)
     \int_0^\infty
       f\bigl(s \bv(1) \bigr)
             \nu_{c^*/\lambda_{1}}(\mathrm{d}s).
\end{align*}
Here $\widetilde{Z}(1)$ satisfies
$e^{-\lambda_{1}t}\widetilde{\bxi}_{\delta}(t,1)\convas
\widetilde{Z}(1)\bv(1)$, and  
\begin{equation}
\label{eq:slope1}
a(1):=\frac{v^{(2)}{(1)}}{v^{(1)}{(1)}} = \frac{\gamma-\alpha+\sqrt{D_0(1)}}{2\gamma\rho_{\bullet 1}}.
\end{equation}
\end{Theorem}

Write
\begin{equation}\label{e:def a(m)}
  a(m):= {v^{(2)}{(m)}}/{v^{(1)}{(m)}}, \quad m=1,\ldots, K.
  \end{equation}
Theorem~\ref{thm:mrv} 
shows that the limiting measure $\mu_1$ concentrates Pareto 
mass on the ray
$\mathcal{L}_{(1)}:=\{(x,y)\in \RR_+^2: y=a{(1)}x\}$ and
concentration on $\mathcal{L}_{(1)}$
suggests that at scale $b(t)=t^{\lambda_{1}/c^*}$,
  large in- and out-degree pairs in the PA model with heterogeneous 
reciprocity levels satisfy $\mathcal{O}\approx a{(1)}\mathcal{I}$.

\begin{proof}
Similar to the proof of Theorem~6 in \cite{wang:resnick:2022}, we first see that
$\PP(\widetilde{Z}_1>0)=1$ by the property of MBI processes.
Also, {since $\lambda_1>\lambda_r$} for $r=2,\ldots,K$, 
\[
e^{-\lambda_{1}t}\widetilde{\bxi}_\delta(t, {r})\convas 0.
\]
Therefore, 
\[
e^{-\lambda_{1}t}\widetilde{\bxi}_\delta(t, L^*)\convas
\widetilde{Z}(1)\bv(1) \ind_{\{L^*=1\}} 
\]
Then the proof of 
 Theorem~\ref{thm:mrv} is an application of Theorem
 \ref{th:extendBrei} after making the identifications
 \begin{align*}
&\bxi (t)=t^{-1} \widetilde{\bxi}_\delta\left(\frac{1}{\lambda_{1}} \log t, L^*\right), 
           &&\bxi_\infty =\widetilde{Z}(1)\bv(1) \ind_{\{L^*=1\}},
   &&X=e^{\lambda_{1} T^*},\\
   &b(t)=t^{\lambda_{1}/c^*},&&
c=c^*/\lambda_{1}. && {}
 \end{align*}
 The remaining piece is to show the moment condition \eqref{e:extraCond} in this
 context and we will show
any $\delta \ge 0$ and any $q=1,2,\ldots$, there exists some constant $C(\delta,q)>0$ such that 
\begin{align}\label{eq:claim_moment}
\sup_{t\ge 0}e^{-\lambda_{1} qt} \EE\left[\left(\widetilde{\xi}^{(1)}_\delta(t,1)\right)^q\right]
\le C(\delta, q),
\end{align}
which is true by Proposition~2 in \cite{wang:resnick:2022}.
\end{proof}

Next, Theorem~\ref{thm:hrv} gives a second {\it hidden\/} regular 
variation (HRV) regime after removing $\mathcal{L}_{(1)}$ 
\citep{DasRes2015, das:resnick:2017, 
resnickbook:2007,das:mitra:resnick:2013,lindskog:resnick:roy:2014}.
The existence of HRV has been detected empirically in network data
\citep{das:resnick:2017}, 
and here we theoretically prove HRV {present} in the PA model with heterogeneous reciprocity.

The limit measure given in Theorem~\ref{thm:mrv} concentrates on 
$\mathcal{L}_{(1)}$.
Thus, we may seek a regular variation property on
$\RR_+^2\setminus \mathcal{L}_{(1)}$ using a weaker scaling function
$b_0(t)$. A
convenient way to seek the hidden regular variation is by using {\it
  generalized polar coordinates\/} which in this case amount to the
  transformation
  $$\bx \to \Bigl(d_1(\bx, \mathcal{L}_{(1)}), \frac{\bx}{d_1(\bx,\mathcal{L}_{(1)})}\Bigr),$$
where $d_1(\bx,\by) $ is a metric on $\RR_+^2\setminus \{\origin\}$
chosen for convenience to be the $L_1$-metric.
The $L_1$-distance of a point $(x,y)$ to the line $\mathcal{L}_{(1)}$ is readily computed
to be
$$
d_1 \bigl((x,y),\mathcal{L}_{(1)} \bigr) ={|y-a(1)x|}/{\text{max}\{1,a(1\}},
$$
{and we use a scaled version }
\begin{equation}\label{e:d'}
d'_1 \bigl((x,y), \mathcal{L}_{(1)}\bigr) =|y-a(1)x|.
\end{equation}
Define $\aleph_{\mathcal{L}_{(1)}} :=\{\bx \in \RR_+^2 \setminus \mathcal{L}_{(1)}: d'_1(\bx, \mathcal{L}_{(1)}
)=1\}$,
which are 2 line segments {in $\RR_+^2$} parallel to $\mathcal{L}_{(1)}$.
 Hidden regular
variation will be present for $(\mathcal{I},\mathcal{O}) \stackrel{d}{=} \widetilde
\bxi_\delta (T^*,L^*)$ if
$$tP\left[ \left(\frac{d'_1(\widetilde \bxi_\delta (T^*,L^*),\mathcal{L}_{(1)} )}{b_0(t)},
\frac{\widetilde \bxi_\delta (T^*,L^*)}{d'_1\left(\widetilde \bxi_\delta
  (T^*,L^*),\mathcal{L}_{(1)} \right)} \right) \in \cdot \,\right]
$$
converges to a limit measure in
$\mathbb{M}((\RR_+\setminus\{0\})\times \aleph_{\mathcal{L}_{(1)}})$. 
The next theorem explains the convergence on $\mathbb{M}(\mathbb{R}_+^2\setminus\mathcal{L}_{(1)})$.

\begin{Theorem}\label{thm:hrv}
Assume that $\lambda_2>\lambda_{1}/2$, and $\lambda_2\ge \log 2$.
Let $\widetilde{Z}(2)$ be the limiting random variable satisfying
$e^{-t\lambda_2}\widetilde{\bxi}_\delta(t,2)\convas \widetilde{Z}(2)\bv(2)$
as $t\to\infty$. Then we have in $\mathbb{M}(\RR^2_+\setminus \mathcal{L}_{(1)})$ that
\begin{align}\label{eq:conv_hrv}
t\PP\left[\frac{(\mathcal{I},\mathcal{O})}{t^{\lambda_2/c^*}}\in\cdot\right]
\to \mu_2,
\end{align}
where the limit measure $\mu_2\in
\mathbb{M}(\mathbb{R}_+^2\setminus\mathcal{L}_{(1)})$ 
concentrates on the ray $y=a(2)x,$ $x>0$
(recall \eqref{e:def a(m)}) in the first quadrant and satisfies 
for $g\in \mathcal{C}(\mathbb{R}_+^2\setminus\mathcal{L}_{(1)})$, 
\begin{align*}
\mu_2(g) = &\PP(L^*=2)\int_0^\infty \EE\left(g(y\widetilde{Z}(2)\bv(2))\right)
             \nu_{c^*/\lambda_2}(\mathrm{d}y).\\
             =&
                \PP(L^*=2)   \EE\left(
                \widetilde{Z}(2)^{c^*/\lambda_2} \right)
                \int_0^\infty g(y\bv(2) )
\nu_{c^*/\lambda_2}(\mathrm{d}y).                
\end{align*}
\end{Theorem}

Theorem~\ref{thm:hrv} suggests that after deleting large in- and
out-degree pairs close to $\mathcal{L}_{(1)}$, at the scale
$t^{\lambda_2/c^*}$
the remaining large
observations of in- and out-degrees tend to concentrate around another
line  
\[
\mathcal{L}_{(2)} := \{(x,y)\in (0,\infty)^2: y=a(2)x\}.
\]
Define also that $\mathcal{L}_{(0)}:= \{\origin\}$, and
combining Theorems~\ref{thm:mrv} and \ref{thm:hrv} gives  
\[
\PP((\mathcal{I},\mathcal{O})\in\cdot)\in 
\bigcap_{m=1}^2\text{MRV}\left(c^*/\lambda_{m}, t^{\lambda_{m}/c^*}, \mu_m, \mathbb{R}_+^2\setminus \left(\bigcup_{i=0}^{m-1}\mathcal{L}_{(i)}\right)\right).
\]
\begin{Remark}
{\rm 
Results in Theorem~\ref{thm:hrv} can be extended as follows. 
Let $\lambda_{m_0}$, $m_0\ge 2$, denote the $m_0$-th largest eigenvalue such that
for all $m\in \{2,\ldots, m_0\}$, 
\[
\lambda_{m}>\lambda_{m-1}/2,\qquad\text{and}\qquad \lambda_{m}\ge \log 2.
\]
For $m=1,\ldots,K$, write $\mathcal{L}_{(m)}:= \{(x,y)\in (0,\infty)^2: y=a(m)x\}$,
and set $\widetilde{Z}(m)$ to be the limiting random variable satisfying
$e^{-t\lambda_{m}}\widetilde{\bxi}_\delta(t,m)\convas \widetilde{Z}(m)\bv(m)$
as $t\to\infty$.
Define also the measure
$\mu_m\in
\mathbb{M}\left(\mathbb{R}_+^2\setminus\bigl(\bigcup_{i=1}^m\mathcal{L}_{(i)}\bigr)\right)$
such that for $g\in
\mathcal{C}\left(\mathbb{R}_+^2\setminus\bigl(\bigcup_{i=1}^m\mathcal{L}_{(i)}\bigr)\right)$, 
\[
\mu_m(g) = \PP(L=m)\int_0^\infty \EE\left(g(y\widetilde{Z}(m)\bv(m\right)
\nu_{c^*/\lambda_{m}}(\mathrm{d}y).
\]
Then applying the proof technique of Theorem~\ref{thm:hrv} for $m_0$ times gives 
\[
\PP \bigl((\mathcal{I},\mathcal{O})\in\cdot\bigr)\in 
\bigcap_{m=1}^{m_0}\text{MRV}\left(c^*/\lambda_{m}, t^{\lambda_{m}/c^*}, \mu_m, \mathbb{R}_+^2\setminus \left(\bigcup_{i=0}^{m-1}\mathcal{L}_{(i)}\right)\right).
\]
}
\end{Remark}
\begin{proof}
  Define $\bm{\theta}_{(2)} := (1,a(2))/\left|a(1)-a(2)\right|)
\in \aleph_{\mathcal{L}_1}$. 
Applying the generalized polar transformation 
shows that verifying \eqref{eq:conv_hrv} is equivalent to justifying
\begin{align}\label{eq:conv_hrv1}
  t \PP\Bigl[   \Bigl(&
  \frac{d_1'   ( \widetilde{\bxi}_\delta(T^*,L^*),\mathcal{L}_1)}
{ t^{\lambda_2/c^*}}, 
  \frac{\widetilde{\bxi}_\delta(T^*,L^*)}{d_1'   (
  \widetilde{\bxi}_\delta(T^*,L^*),\mathcal{L}_1) } 
  \Bigr) \in \cdot\,\Bigr]\nonumber \\
 =& t\PP\left[\left(
\frac{\left|a{(1)} \widetilde{\xi}^{(1)}_\delta(T^*,L^*)-\widetilde{\xi}^{(2)}_\delta(T^*,L^*)\right|}{t^{\lambda_2/c^*}}, 
\frac{\widetilde{\bxi}_\delta(T^*,L^*)}{\left|a{(1)}
        \widetilde{\xi}^{(1)}_\delta(T^*,L^*)-\widetilde{\xi}^{(2)}_\delta(T^*,L^*)\right|}\right)\in\cdot\right]\nonumber\\ 
\to                      & C_2\nu_{c^*/\lambda_2}(\cdot)\times 
\epsilon_{\bm{\theta}_{(2)}}(\cdot),
\end{align}
in $\mathbb{M}\left((0,\infty)\times\aleph_{[\mathcal{L}_{(1)}]}\right)$, where
\begin{align*}
C_2 = \PP(L^*=2)\frac{c^*}{\lambda_2}
&\times\int_0^\infty z^{-1-c^*/\lambda_2}\PP\left(\widetilde{Z}(2)>
\frac{1/z}{v^{(1)}{(2)} \left|a(1)-a(2)\right|}\right)\mathrm{d}z.
\end{align*}

To prove \eqref{eq:conv_hrv1}, we first claim that for
$\lambda_2>\lambda_{1}/2$, as $t\to\infty$,  
\begin{align}\label{eq:hrv_claim}
    e^{-t\lambda_2} d_1'(\widetilde{\bxi}_\delta(T^*,L^*),\mathcal{L}_1)=&
e^{-t\lambda_2}\left|a(1)\widetilde{\xi}^{(1)}_{\delta}(t,L^*)-
  \widetilde{\xi}^{(2)}_{\delta}(t,L^*)\right|\nonumber \\
  \convas &
\left|a(1)-a(2)\right| v^{(1)}{(2)}\widetilde{Z}(2)
\bm{1}_{\{L^*=2\}}.
\end{align}
In addition, since
\begin{align*}
\frac{\widetilde{\bxi}_{\delta}(t,m)}{\left|a(1)\widetilde{\xi}^{(1)}_{\delta}(t,m)-\widetilde{\xi}^{(2)}_{\delta}(t,m)\right|}
\convas \frac{(1, a(m))}{\bigl|a(1)-a(m)\bigr|},
\end{align*}
 we have
\begin{align}
&\left(\frac{\left|a(1)\widetilde{\xi}^{(1)}_{\delta}(t,L^*)-\widetilde{\xi}^{(2)}_{\delta}(t,L^*)\right|}{e^{t\lambda_2 }},\,
 \frac{\widetilde{\bxi}_{\delta}(t,L^*)}{\left|a(1)\widetilde{\xi}^{(1)}_{\delta}(t,L^*)-\widetilde{\xi}^{(2)}_{\delta}(t,L^*)\right|}\right)\nonumber\\
 &= \sum_{m=1}^K \left(\frac{\left|a(1)\widetilde{\xi}^{(1)}_{\delta}(t,m)-\widetilde{\xi}^{(2)}_{\delta}(t,m)\right|}{e^{t\lambda_2 }},\,
 \frac{\widetilde{\bxi}_{\delta}(t,m)}{\left|a(1)\widetilde{\xi}^{(1)}_{\delta}(t,m)-\widetilde{\xi}^{(2)}_{\delta}(t,m)\right|}\right)\ind_{\{L^*=m\}}\nonumber
 \\
&\quad\convas \left(
\bigl|a(1)-a(2)\bigr|v^{(1)}{(2)}\widetilde{Z}(2)\bm{1}_{\{L^*=2\}},\,
\sum_{m=1}^K \frac{(1, a(m))}{\bigl|a(1)-a(m)\bigr|}\bm{1}_{\{L^*=m\}}
\right).\label{eq:MBI}
\end{align}
Equation \eqref{eq:MBI} further gives that in 
$\mathbb{M}\left((0,\infty)\times\left((0,\infty)\times\aleph_{[\mathcal{L}_{(1)}]}\right)\right)$,
\begin{align*}
&t\PP\left[
\left(\frac{e^{T^*}}{t^{1/c^*}},\right.\right.\\
&\left.\left.\sum_{m=1}^K\left(\frac{\left|a(1)\widetilde{\xi}^{(1)}_{\delta}(T^*,m)-\widetilde{\xi}^{(2)}_{\delta}(T^*,m)\right|}{e^{\lambda_2T^*}},
\frac{\widetilde{\bxi}_{\delta}(T^*,m)}{\left|a(1)\widetilde{\xi}^{(1)}_{\delta}(T^*,m)-\widetilde{\xi}^{(2)}_{\delta}(T^*,m)\right|}\right)
\bm{1}_{\{L^*=m\}}\right)
\in\cdot\right]\\
&\longrightarrow \PP(L^*= 2)\, \nu_{c^*/\lambda_2}(\cdot)\times\PP\left(\bigl|a(1)-a(2)\bigr|v^{(1)}{(2)}\widetilde{Z}(2)\in\cdot\right)\times \epsilon_{\bm{\theta}_{(2)}}(\cdot).
\end{align*}
Hence, as long as we check the moment condition that for $q=1,2,\ldots$ and $\delta>0$, 
\begin{align}
\label{eq:moment}
\sup_{t\ge 0} e^{-qt\lambda_2}\EE\left[\left|a(1)\widetilde{\xi}^{(1)}_{\delta}(t,2)-\widetilde{\xi}^{(2)}_{\delta}(t,2)\right|^q\right]<\infty,
\end{align}
then applying the generalized Breiman's theorem (cf. Theorem~\ref{th:extendBrei}) completes the proof of \eqref{eq:conv_hrv1}.
To verify \eqref{eq:moment}, we notice that
since 
$$
\left|a(1)\widetilde{\xi}^{(1)}_{\delta}(t,2)-\widetilde{\xi}^{(2)}_{\delta}(t,2)\right|\le a(1)\widetilde{\xi}^{(1)}_{\delta}(t,2)+\widetilde{\xi}^{(2)}_{\delta}(t,2),
$$
it suffices to show 
\[
\sup_{t\ge 0} e^{-qt\lambda_2}\EE\left[\left(a(1)\widetilde{\xi}^{(1)}_{\delta}(t,2)+\widetilde{\xi}^{(2)}_{\delta}(t,2)\right)^q\right]<\infty,
\qquad \lambda_2\ge\log 2,
\]
which is true by Proposition 2 in \cite{wang:resnick:2022}.

It remains to prove the claim in \eqref{eq:hrv_claim}.
Since $\left|a(m)\widetilde{\xi}^{(1)}_\delta(t,m)-\widetilde{\xi}^{(2)}_\delta(t,m)\right|\le a(m)\widetilde{\xi}^{(1)}_\delta(t,m)+\widetilde{\xi}^{(2)}_\delta(t,m)$, then for $m\notin \{1, 2\}$, we have
\[
e^{-t\lambda_2}\left|a(m)\widetilde{\xi}^{(1)}_\delta(t,m)-\widetilde{\xi}^{(2)}_\delta(t,m)\right|\convas 0,
\qquad \text{as }t\to\infty.
\]
Hence, we only need to consider $m=1, 2$.
Consider the case when $m=1$, and define $\{\bm{\psi}_i(t, 1):t\ge 0\}_{i\ge 1}$ as a sequence of iid two-type branching processes (without immigration) with group label $1$, whose branching structure is specified through $A_{1}$. Also, assume that $\bm{\psi}_i(0, 1)$ is a 2-dimensional random vector with distribution $p_0(\br, 1)$ (cf. \eqref{eq:def_p0}).
Let $0<\tau_1< \tau_2<\ldots$ be arrival times of points in a homogeneous Poisson process with rate $\delta>0$, which is independent from $\{\bm{\psi}_i(t, 1):t\ge 0\}_{i\ge 1}$.
Then by the distributional construction of the MBI process in \cite{rabehasaina:2021},
we write 
\begin{align}
\label{eq:MBI_dist}
\widetilde{\bxi}_\delta(t,1) 
= \sum_{i=1}^\infty \bm{\psi}_i(t-\tau_i, 1)\ind_{\{t\ge \tau_i\}}.
\end{align}

Next, we use a Borel Cantelli argument to show
\begin{align}
\label{eq:pfL1}
e^{-t\lambda_2}\left|a(1)\widetilde{\xi}^{(1)}_\delta(t,1)-\widetilde{\xi}^{(2)}_\delta(t,1)\right|\convas 0,
\end{align}
as $t\to\infty$, i.e. we will show that
\begin{align}
\label{eq:pfBC}
\sum_{n=1}^\infty\PP\left(e^{-n\lambda_2}\left|a(1)\widetilde{\xi}^{(1)}_\delta(n,1)-\widetilde{\xi}^{(2)}_\delta(n,1)\right|>\epsilon\right)<\infty.
\end{align}
By Markov's inequality, we have for $n\ge 1$, 
\begin{align}
\PP&\left(e^{-n\lambda_2}\left|a(1)\widetilde{\xi}^{(1)}_\delta(n,1)-\widetilde{\xi}^{(2)}_\delta(n,1)\right|>\epsilon\right)\nonumber\\
&\le \epsilon^{-2}\EE\left[e^{-2n\lambda_2}\left|a(1)\widetilde{\xi}^{(1)}_\delta(n,1)-\widetilde{\xi}^{(2)}_\delta(n,1)\right|^2\right]\nonumber\\  
&\le  \epsilon^{-2} e^{-n(2\lambda_2-\lambda_1)}
\sup_{t\ge 0} \EE\left[\left(e^{-t\lambda_1/2}\left(a(1)\widetilde{\xi}^{(1)}_\delta(t,1)-\widetilde{\xi}^{(2)}_\delta(t,1\right)\right)^2\right].
\label{eq:bound_mbi}
\end{align}
Note that the vector $[a(1), -1]^T$ is the right eigenvector associated with the smaller eigenvalue of $A_{1}$. 
Recall the discussion in Section~\ref{subsec:MBI} that
 $\lambda'_1$ is the smaller eigenvalue of $A_{1}$, then 
the condition $\lambda_1\ge \log 2$ guarantees that $\lambda_1>2\lambda'_1$.
Hence, we apply Theorem~V.7.1(i) in \cite{athreya:ney:1972} to conclude that for $i\ge 1$,
\[
\sup_{t\ge 0} \EE\left[\left(e^{-t\lambda_1/2}\left(a(1)\psi^{(1)}_i(t,1)-\psi^{(2)}_i(t,1\right)\right)^2\right]<\infty.
\]
We then conclude from \eqref{eq:MBI_dist} that
\begin{align*}
&\left(e^{-t\lambda_1/2}\left(a(1)\widetilde{\xi}^{(1)}_\delta(t,1)-\widetilde{\xi}^{(2)}_\delta(t,1\right)\right)^2\\
&= \left(e^{-t\lambda_1/2}\sum_{i=1}^\infty\left(a(1)\psi_i^{(1)}(t-\tau_i,1)-\psi_i^{(2)}(t-\tau_i,1\right)\ind_{\{t\ge \tau_i\}}\right)^2\\
&\le \left(\sum_{i=1}^\infty e^{-\tau_i\lambda_1/2}\right)\\
&\quad\times\left(\sum_{i=1}^\infty e^{-\tau_i\lambda_1/2}\left(e^{-(t-\tau_i)\lambda_1/2}\left(a(1)\psi_i^{(1)}(t-\tau_i,1)-\psi_i^{(2)}(t-\tau_i,1)\right)\ind_{\{t\ge \tau_i\}}\right)^2\right).
\end{align*}
Therefore,
\begin{align*}
\EE&\left(e^{-t\lambda_1/2}\left(a(1)\widetilde{\xi}^{(1)}_\delta(t,1)-\widetilde{\xi}^{(2)}_\delta(t,1\right)\right)^2\\
&\le \sup_{t\ge 0} \EE\left[\left(e^{-t\lambda_1/2}\left(a(1)\psi^{(1)}_i(t,1)-\psi^{(2)}_i(t,1\right)\right)^2\right]\\
&\quad \times \EE\left[\sum_{i=1}^\infty e^{-\tau_i\lambda_1/2}\left(\sum_{i=1}^\infty e^{-\tau_i\lambda_1/2}\right)\right]<\infty.
\end{align*}
Also, since we assume $\lambda_2>\lambda_1/2$, then 
\eqref{eq:pfBC} follows from \eqref{eq:bound_mbi},
leading to \eqref{eq:pfL1}.
Then the claim in \eqref{eq:hrv_claim} follows by realizing that 
\begin{align*}
e^{-t\lambda_2}\left(a(1)\widetilde{\xi}^{(1)}_{\delta}(t,2)-
\widetilde{\xi}^{(2)}_{\delta}(t,2)\right)
&\convas \left(a(1)v^{(1)}(2)-v^{(2)}(2)\right) \widetilde{Z}(2)\\
&= \left(a(1)-a(2)\right) v^{(1)}(2) \widetilde{Z}(2).
\end{align*}
\end{proof}

\section{Concluding Remarks}\label{sec:conclusion}

In this paper, we propose a preferential attachment model with heterogeneous reciprocity levels and study its theoretical properties. Using the MBI embedding technique, we find that the distribution of large in- and out-degrees is jointly regularly varying, concentrating along a specific ray. Additionally, after deleting large in- and out-degree pairs close to the ray, we further have hidden regular variation with limit measure concentrating around another ray. 

We now outline some open problems related to this model, which will be left as future research.
\subparagraph*{Estimation.} 
The fitting of the proposed model remains open at this point. For a given dataset, we need to first decide how many communication groups ($K$) should be assumed. One may consult classical clustering methods such as $k$-means and $k$-nearest neighbors to determine a proper $K$ beforehand, but their applications to the network framework needs rigorous justification. 
Once $K$ is chosen, we can also derive proper tools (similar to those in \cite{das:resnick:2017}) to detect the existence of hidden regular variation under the network setup. 

\subparagraph*{Varying group labels.} 
So far we have assumed that the communication group of a user (node) is determined upon its creation and remains unchanged afterwards. For real-world applications, however, this may be a naive assumption, since users' interaction patterns can change over time. A possible extension is to assume that the group label for each node follows another Markov chain with finite state space. One may consider applying variational Bayesian inference methods (cf. \cite{matias:miele:2017}) to these extended models. 

\subparagraph*{Covariate-dependent reciprocity probabilities.}
The model studied in this paper assumes a deterministic matrix, $\boldsymbol{\rho}$, to characterize the reciprocation between different communication groups. 
If more node-specific information (covariates) is available, e.g. various demographic, socio-economic, and behavioral factors, then we can generalize the definition of the matrix $\boldsymbol{\rho}$ to be covariate-dependent. This may possibly lead to heterogeneous extremal behaviors. 

\section{Proofs of Results in Section~\ref{sec:edge}}\label{sec:pf}

\subsection{Proof of Lemma~\ref{lem:xy}}\label{subsec:pf_lem1}

Define a function $\boldf: \mathcal{Z} \mapsto \mathcal{Z}$ such that
for $i =1,\ldots, K$,
\begin{align*}
f_i(\bz) = \alpha \frac{z_i+\delta \pi_i}{\sum_{j=1}^K z_j +\delta} 
+ \gamma \rho_{\bullet i} \frac{z_{K+i}+\delta \pi_i}{\sum_{j=1}^K z_{K+j} +\delta} 
+\gamma \pi_i + \alpha \pi_i \sum_{r=1}^K \rho_{r, i} \frac{z_i+\delta \pi_i}{\sum_{j=1}^K z_j +\delta},
\end{align*}
and for $i =K+1,\ldots, 2K$,
\begin{align*}
f_i(\bz) =& \gamma \frac{z_i+\delta \pi_{i-K}}{\sum_{j=K+1}^{2K} z_j +\delta} 
+ \alpha \rho_{(i-K)\,\bullet} \frac{z_{i-K}+\delta \pi_{i-K}}{\sum_{j=1}^K z_{j} +\delta} 
+\alpha \pi_{i-K} \\
&+ \gamma\pi_{i-K} \sum_{r=1}^K \rho_{i, r} \frac{z_i+\delta \pi_{i-K}}{\sum_{j=K+1}^{2K} z_j +\delta}.
\end{align*}
Let $J(\bz')$ be the Jacobian matrix evaluated at some $\bz'$ between $\bz_1$ and $\bz_2$.
Use $\|\bz\|_1$ to denote the $L_1$-norm for some $\bz\in \mathcal{Z}$, and set $\|\cdot\|_{(1,\infty)}$ to be the $L_{1,\infty}$ norm of a matrix.
Then by the mean value theorem, we see that for 
$\bz_1,\bz_2\in \mathcal{Z}$, 
\begin{align}
\|\boldf(\bz_1)-\boldf(\bz_2)\|_1
&\le \sup_{\bz'\in \mathcal{Z}}\|J(\bz')\|_{(1,\infty)}\|\bz_1-\bz_2\|_\infty\nonumber\\
&\le \sup_{\bz'\in \mathcal{Z}}\|J(\bz')\|_{(1,\infty)}\|\bz_1-\bz_2\|_1.
\label{eq:contraction}
\end{align}
If we can show 
$\sup_{\bz'\in \mathcal{Z}}\|J(\bz')\|_{(1,\infty)}<1$, then by the contraction mapping theorem (cf. Theorem~1.2.2 in \cite{Kirk:2001})
we are able to show the existence of a unique solution to $\boldf(\bz)=\bz$.

To find an upper bound for $\sup_{\bz\in \mathcal{Z}}\|J(\bz)\|_{(1,\infty)}$, we now give upper bounds for the absolute value of each entry in $J(\bz)$.
Let $J_{i,j}(\bz)$ be the $(i,j)$-th entry in $J(\bz)$, and we have
for $1\le i,j \le K$,
\begin{align}\label{eq:Jbound1}
\left|J_{i,j}(\bz)\right| &= \left|\frac{\partial f_i}{\partial z_j}(\bz) \right|
\le \alpha(\pi_i \rho_{j, i}+\ind_{\{i=j\}})\frac{1}{\sum_{m=1}^K z_m +\delta}
\le \frac{\alpha(\pi_i \rho_{j, i}+\ind_{\{i=j\}})}{1 +\delta}.
\end{align}
For $1\le i\le K$ and $K+1\le j\le 2K$,
\begin{align}\label{eq:Jbound2}
\left|J_{i,j}(\bz)\right| &= \left|\frac{\partial f_i}{\partial z_j}(\bz) \right|
= \gamma \rho_{\bullet i} \left|-\frac{z_{K+i}+\delta\pi_i}{\left(\sum_{m=K+1}^{2K}z_m +\delta\right)^2}+\frac{\ind_{\{j=K+i\}}}{\sum_{m=K+1}^{2K}z_m +\delta}\right|
\le \frac{\gamma \rho_{\bullet i}}{1+\delta}.
\end{align}
For $K+1\le i\le 2K$ and $1\le j\le K$, we see that
\begin{align}\label{eq:Jbound3}
\left|J_{i,j}(\bz)\right| &= \left|\frac{\partial f_i}{\partial z_j}(\bz) \right|
= \alpha \rho_{i\bullet} \left|-\frac{z_{i-K}+\delta\pi_{i-K}}{\left(\sum_{m=1}^{K}z_m +\delta\right)^2}+\frac{\ind_{\{j=i-K\}}}{\sum_{m=1}^{K}z_m +\delta}\right|
\le \frac{\alpha \rho_{i\bullet}}{1+\delta}.
\end{align}
Also, for $K+1\le i,j \le 2K$, we have
\begin{align}\label{eq:Jbound4}
\left|J_{i,j}(\bz)\right| &= \left|\frac{\partial f_i}{\partial z_j}(\bz) \right|
\le \gamma(\pi_{i-K} \rho_{i, j}+\ind_{\{i=j\}})\frac{1}{\sum_{m=K+1}^{2K} z_m +\delta}
\le \frac{\gamma(\pi_{i-K} \rho_{i, j}+\ind_{\{i=j\}})}{1 +\delta}.
\end{align}
Let $J^*_{i,j}$ be the $(i,j)$-th entry of the $J^*$ matrix.
Equations \eqref{eq:Jbound1}--\eqref{eq:Jbound4} imply that 
\begin{align*}
\sup_{\bz'\in \mathcal{Z}}\|J(\bz')\|_{(1,\infty)}&
\le \frac{1}{1+\delta}\bigvee_{j=1}^K \left(\sum_{i=1}^K|J^*_{i,j}|\right) = \frac{1}{1+\delta}\|J^*\|_1.
\end{align*}
Therefore, as long as
\[
\delta> \|J^*\|_1-1,
\]
Equation~\eqref{eq:contraction} gives 
\[
\|\boldf(\bz_1)-\boldf(\bz_2)\|_1
\le \frac{1}{1+\delta}\|J^*\|_1 \|\bz_1-\bz_2\|_1 < \|\bz_1-\bz_2\|_1,
\]
indicating that $\boldf(\bz)=\bz$ has a unique solution in $\mathcal{Z}$.

\subsection{Proof of Theorem~\ref{thm:En}}\label{subsec:pf_thmEn}

We first show the concentration of $|E(n)|/n$ around its expectation, $\EE[|E(n)|]/n$.
Suppose we have a graph $G(n) = (V(n),E(n))$, constructed from the PA model with heterogeneous reciprocity levels. We claim that for $C_0>2$,
\begin{align}
\label{eq:concenE}
\PP\left(\bigl||E(n)|-\EE[|E(n)|]\bigr|\ge C_0\sqrt{n\log n}\right) \le 2n^{-2}.
\end{align}
We prove \eqref{eq:concenE} using the Azuma-Hoeffding inequality.
For $k\le n$, define $M_k := \EE\bigl(|E(n)|\big\vert G(k)\bigr)$, then $\EE(M_{k+1}|G(k))=M_k$.
So we need to consider the martingale difference:
\[
M_{k+1}-M_k=\EE\bigl(|E(n)|\big\vert G(k+1)\bigr)-\EE\bigl(|E(n)|\big\vert G(k)\bigr).
\]
Define another graph $G'(n)\equiv (V'(n), E'(n))$ such that $G'(s) = G(s)$ for $s\le k$, while $G'(s)$ evolves independently of $\{G(s)\}_{s\ge k+1}$ for $s\ge k+1$, according to the same evolution rules as in Section~\ref{sec:model}. Then we have
\begin{align*}
M_{k+1}-M_k &= \EE\bigl(|E(n)|\big\vert G(k+1)\bigr)-\EE\bigl(|E'(n)|\big\vert G(k+1)\bigr)\\
&= \EE\left[\EE\bigl(|E(n)|-|E'(n)|\big\vert G(k+1), G'(k+1)\bigr)\middle\vert G(k+1)\right]\\
&= \EE\left[|E(k+1)|-|E'(k+1)|\middle\vert G(k+1)\right],
\end{align*}
which gives
\[
|M_{k+1}-M_k|\le 1. 
\]
Since $|E(n)|-\EE[|E(n)|] = \sum_{k=0}^{n-1}(M_{k+1}-M_k)$, then applying the Azuma-Hoeffding inequality gives that for $\epsilon>0$,
\[
\PP\left(\bigl||E(n)|-\EE[|E(n)|]\bigr|\ge \epsilon\right) \le 2 e^{-\frac{\epsilon^2}{2n}},
\]
and \eqref{eq:concenE} follows by setting $\epsilon= C_0\sqrt{n\log n}$ for $C_0>2$.

Next, we study the convergence of $\EE[|E(n)|]/n$.
Since $x_m$ satisfies \eqref{eq:xm},
then we see from \eqref{eq:Zin} that
\begin{align*}
\EE^{\mathcal{G}_n}(\Zin_m(n+1)) 
&= \Zin_m(n) + \left(\EE^{\mathcal{G}_n}\left(|\Ein_m(n+1)|-|\Ein_m(n)|\right)-x_m\right)\\
&= \Zin_m(n) + \alpha\left(\frac{|\Ein_m(n)|+\delta|V_m(n)|}{|E(n)|+\delta |V(n)|}-\frac{x_m+\delta\pi_m}{\sum_r x_r +\delta}\right)\\
&+ \gamma\rho_{\bullet m}\left(\frac{|\Eout_m(n)|+\delta|V_m(n)|}{|E(n)|+\delta |V(n)|}- \frac{y_m+\delta\pi_m}{\sum_r y_r +\delta}\right)\\
&+ \alpha\pi_m \sum_r \rho_{r, m} \left(\frac{|\Ein_r(n)|+\delta|V_r(n)|}{|E(n)|+\delta |V(n)|}-\frac{x_r+\delta\pi_r}{\sum_r x_r +\delta}\right),
\end{align*}
which implies
\begin{align}
\left|\EE(\Zin_m(n+1))\right| &\le \left|\EE(\Zin_m(n))\right| + \alpha\left|\EE\left(\frac{|\Ein_m(n)|+\delta|V_m(n)|}{|E(n)|+\delta |V(n)|}-\frac{x_m+\delta\pi_m}{\sum_r x_r +\delta}\right)\right|\nonumber\\
&+ \gamma\rho_{\bullet m}\left|\EE\left(\frac{|\Eout_m(n)|+\delta|V_m(n)|}{|E(n)|+\delta |V(n)|}- \frac{y_m+\delta\pi_m}{\sum_r y_r +\delta}\right)\right|\nonumber\\
&+ \alpha\pi_m \sum_r \rho_{r, m} \left|\EE\left(\frac{|\Ein_r(n)|+\delta|V_r(n)|}{|E(n)|+\delta |V(n)|}-\frac{x_r+\delta\pi_r}{\sum_r x_r +\delta}\right)\right|.
\label{eq:EZin}
\end{align}
Also, we have that for $m=1,\ldots, K$,
\begin{align*}
&\left|\EE\left(\frac{|\Ein_m(n)|+\delta|V_m(n)|}{|E(n)|+\delta |V(n)|}-\frac{x_m+\delta\pi_m}{\sum_r x_r +\delta}\right)\right|\\
&= \left|\EE\left(\frac{\Zin_m(n)}{|E(n)|+\delta |V(n)|}-\frac{x_m+\delta\pi_m}{\sum_r x_r+\delta}\frac{\Delta(n)}{|E(n)|+\delta|V(n)|}\right)\right|\\
&\le \frac{|\EE(\Zin_m(n))|}{(1+\delta)n} 
+ \frac{x_m+\delta\pi_m}{\sum_r x_r+\delta}\frac{|\EE(\Delta(n))|}{(1+\delta)n}. 
\end{align*}
Similarly, we see that 
\begin{align*}
\left|\EE\left(\frac{|\Eout_m(n)|+\delta|V_m(n)|}{|E(n)|+\delta |V(n)|}-\frac{y_m+\delta\pi_m}{\sum_r y_r +\delta}\right)\right|
\le \frac{|\EE(\Zout_m(n))|}{(1+\delta)n} 
+ \frac{y_m+\delta\pi_m}{\sum_r y_r+\delta}\frac{|\EE(\Delta(n))|}{(1+\delta)n}. 
\end{align*}
Therefore, it follows from \eqref{eq:EZin} that
\begin{align*}
\left|\EE(\Zin_m(n+1))\right|&\le 
\left|\EE(\Zin_m(n))\right|\left(1+\frac{\alpha}{(1+\delta)n}\right) + 
\left|\EE(\Zout_m(n))\right|\frac{\gamma\rho_{\bullet m}}{(1+\delta)n}+\\
&+ \frac{\alpha\pi_m}{(1+\delta)n}\sum_r\rho_{r, m}\left|\EE(\Zin_r(n))\right|\\
&+ \frac{\left|\EE(\Delta(n))\right|}{(1+\delta)n} \left(\alpha\frac{x_m+\delta\pi_m}{\sum_r x_r +\delta}+\gamma\rho_{\bullet m}\frac{y_m+\delta\pi_m}{\sum_r y_r+\delta}+ \alpha\pi_m\sum_r\rho_{r, m}\frac{x_r+\delta\pi_r}{\sum_r x_r +\delta}\right).
\end{align*}
Summing over $m$ gives
\begin{align}
\sum_m \left|\EE(\Zin_m(n+1))\right|\le& 
\sum_m \left|\EE(\Zin_m(n))\right|\left(1+\frac{\alpha}{(1+\delta)n}\right) + 
\sum_m\left|\EE(\Zout_m(n))\right|\frac{\gamma\bigvee_m\rho_{\bullet m}}{(1+\delta)n}+\nonumber\\
&+ \frac{\alpha\bigvee_m \rho_{m\bullet}}{(1+\delta)n}\sum_m\left|\EE(\Zin_m(n))\right|
+ \frac{\left|\EE(\Delta(n))\right|}{(1+\delta)n} \left(\alpha+C_\delta\right)\nonumber\\
=& \sum_m \left|\EE(\Zin_m(n))\right|\left(1+\frac{\alpha}{(1+\delta)n}\left(1+\bigvee_m\rho_{m\bullet}\right)\right) \nonumber\\
&+ \sum_m\left|\EE(\Zout_m(n))\right|\frac{\gamma\bigvee_m\rho_{\bullet m}}{(1+\delta)n}+ \frac{\left|\EE(\Delta(n))\right|}{(1+\delta)n} \left(\alpha+C_\delta\right).
\label{eq:sum_Zin}
\end{align}

Following a similar reasoning, we have
\begin{align}
\sum_m \left|\EE(\Zout_m(n+1))\right|\le& 
\sum_m \left|\EE(\Zout_m(n))\right|\left(1+\frac{\gamma}{(1+\delta)n}\left(1+\bigvee_m\rho_{\bullet m}\right)\right) \nonumber\\
&+ \sum_m\left|\EE(\Zin_m(n))\right|\frac{\alpha\bigvee_m\rho_{m\bullet }}{(1+\delta)n}+ \frac{\left|\EE(\Delta(n))\right|}{(1+\delta)n} \left(\gamma+C_\delta\right),
\label{eq:sum_Zout}
\end{align}
and 
\begin{align}
\sum_m& \left|\EE(\Delta(n+1))\right|\\
\le& 
\left|\EE(\Delta(n))\right|\left(1+
\frac{\alpha}{(1+\delta)n}\sum_m\rho_{m\bullet}\frac{x_m+\delta\pi_m}{\sum_r x_r+\delta}+
\frac{\gamma}{(1+\delta)n}\sum_m\rho_{\bullet m}\frac{y_m+\delta\pi_m}{\sum_r y_r+\delta}\right) 
\nonumber\\
&+ \sum_m\left|\EE(\Zin_m(n))\right|\frac{\alpha\bigvee_m\rho_{m\bullet }}{(1+\delta)n}
+ \sum_m\left|\EE(\Zout_m(n))\right|\frac{\gamma\bigvee_m\rho_{\bullet m}}{(1+\delta)n}.
\label{eq:sumZ}
\end{align}
Combining Equations~\eqref{eq:sum_Zin}, \eqref{eq:sum_Zout} and \eqref{eq:sumZ} gives:
\begin{align*}
\begin{bmatrix}
\sum_m|\EE(\Zin_m(n+1))|\\
\sum_m|\EE(\Zout_m(n+1))|\\
|\EE(\Delta(n+1))|
\end{bmatrix}
\le \left(\bI+\frac{1}{n}\bH\right)
\begin{bmatrix}
\sum_m|\EE(\Zin_m(n))|\\
\sum_m|\EE(\Zout_m(n))|\\
|\EE(\Delta(n))|
\end{bmatrix}.
\end{align*}
Then by the Perron-Frobenius theorem, we see that 
\[
\frac{1}{n}\sum_m|\EE(\Zin_m(n))|\to 0, \quad 
\frac{1}{n}\sum_m|\EE(\Zout_m(n))|\to 0, \quad
\text{and}\quad
\frac{1}{n}|\EE(\Delta(n))|\to 0, 
\]
if the largest eigenvalue of $\bH$, $\lambda_H$, is less than 1,
thus completing the proof of the theorem.

\appendix
\section{Generalized Breiman's Theorem}

\begin{Theorem}\label{th:extendBrei}
Suppose $\{\bxi(t): t\geq 0\}$ is an $\RR_+^p$-valued stochastic
process 
for some $p\geq 1$.  Let $X$ be a positive random variable with
regularly varying distribution satisfying for some scaling function $b(t)$,
$$\lim_{t\to\infty} t\PP( X/b(t) >x) =x^{-c} =:\nu_c \bigl((x,\infty)\bigr), \quad x>0, c>0.$$
Further suppose
\begin{enumerate}
\item For some finite and positive random vector $\bxi_\infty$,
  $$\lim_{ t \to \infty} {\bxi(t)} =\bxi_\infty \quad (\text{almost surely});$$
  \item The random variable $X$ and the process $\bxi (\cdot)$ are independent.
  \end{enumerate}
  Then:

  (i) In $\mathbb{M}(\RR_+^p \times (\RR_+\setminus \{0\}))$,
  \begin{equation}\label{e:beforeMult}
    t\PP\Bigl[ \Bigl({\bxi(X)}, \frac{X}{b(t)}\Bigr) \in \cdot \,\Bigr]
    \longrightarrow \PP(\bxi_\infty \in \cdot \,) \times
    \nu_c (\cdot)=:\eta(\cdot).\end{equation}
  If $\bxi_\infty$ is of the form $\bxi_\infty =:L\bv$ where $L>0$
  almost surely and $\bv \in (0,\infty)^p$, then $\eta (\cdot)$
  concentrates on the subcone $\mathcal{L}\times (\RR_+\setminus
  \{0\})$ where $\mathcal{L}=\{\theta \bv : \theta>0\}$.

  (ii) If additionally, for some $c'>c$ we have the condition
  \begin{equation}\label{e:extraCond}
  \kappa:=  \sup_{t\geq 0}   \EE\left[ \Bigl(  { \|\bxi(t) \|}  \Bigr)^{c'}\right]
    <\infty,
  \end{equation}
  for some $L_p$ norm $\|\cdot\|$,  then the product 
  of components in
  \eqref{e:beforeMult}, $\bxi(X)X $, has a regularly varying distribution with
  scaling function $b(t)$ and 
    in $\mathbb{M}(\RR_+^p \setminus
\{\origin\})$, 
  \begin{equation}\label{e:prodOK}
    t\PP\Bigl[\frac{X\bxi(X)}{b(t)} \in \cdot \,\Bigr]
    \longrightarrow \left(\PP(\bxi_\infty \in \cdot \,) \times \nu_c\right) \circ h^{-1},
  \end{equation}
where $h(\by,x)=x\by$.
\end{Theorem}

For the classical Breiman Theorem where $p=1$ and $\bxi(t)\equiv \bxi_\infty$,
 \eqref{e:extraCond} is the expected moment condition.

\bibliographystyle{imsart-number}
\bibliography{bibfile_recip.bib}
\end{document}